\documentclass[10pt,journal,compsoc]{IEEEtran}

\usepackage{subcaption}
\usepackage{amsthm}
\newtheorem{lemma}{Lemma}
\newtheorem{assumption}{Assumption}
\newtheorem{definition}{Deﬁnition}
\newtheorem{theorem}{Theorem}
\usepackage{cite}
\usepackage{amsmath,amssymb,amsfonts}
\usepackage{algorithmic}
\usepackage{graphicx}
\usepackage{array}
\usepackage{multirow}
\usepackage{textcomp}
\usepackage{xcolor}
\usepackage{booktabs} 
\usepackage{bm}
\usepackage{ragged2e}
\definecolor{darkgreen}{RGB}{0,100,0}
  
\newcommand{\change}{\textcolor{black}}

\newcommand{\ptitle}[1]{\textit{\textbf{#1}}}
\graphicspath{{./fig/}}
\usepackage{stfloats}
\usepackage{algorithm,algorithmic} 
\def\BibTeX{{\rm B\kern-.05em{\sc i\kern-.025em b}\kern-.08em
    T\kern-.1667em\lower.7ex\hbox{E}\kern-.125emX}}

\begin{document}


\title{Importance Ranking in Complex Networks via Influence-aware Causal Node Embedding}

\author{Jiahui Gao,
        Kuang Zhou*,
        Yuchen Zhu and Keyu Wu
\IEEEcompsocitemizethanks{\IEEEcompsocthanksitem *Corresponding author. 
\IEEEcompsocthanksitem Jiahui Gao, Kuang Zhou and Yuchen Zhu are with the School of Mathematics and Statistics, Northwestern Polytechnical University, Xi’an, China and MOE Key Laboratory for Complexity Science in Aerospace, Northwestern Polytechnical University, Xi’an, China (email: gaojiahui\_0425@mail.nwpu.edu.cn; kzhoumath@nwpu.edu.cn; h1nkik@mail.nwpu.edu.cn).
\IEEEcompsocthanksitem Keyu Wu is with the College of Systems Engineering, National University of Defense Technology, Changsha, China (keyuwu@nudt.edu.cn).
\IEEEcompsocthanksitem This work was supported by the National Social Science Fund of China (No.~23BTJ053), the Natural Science Basic Research Plan in Shaanxi Province of China (No.~2025JC-YBMS-678), the National Natural Science Foundation of China (No.~92371101, 62273352), and the Practice and Innovation Funds for Graduate Students of Northwestern Polytechnical University (No.~PF2025074).}
}

\markboth{Journal of \LaTeX\ Class Files,~Vol.~14, No.~8, August~2015}%
{Shell \MakeLowercase{\textit{et al.}}: Bare Demo of IEEEtran.cls for Computer Society Journals}

\IEEEtitleabstractindextext{%
\begin{abstract}
\justifying
Understanding and quantifying node importance is a fundamental problem in network science and engineering, underpinning a wide range of applications 
such as influence maximization, social recommendation, and network dismantling. Prior research often relies on centrality measures or advanced graph embedding techniques 
using structural information, followed by downstream classification or regression tasks to 
identify critical nodes. 
However, these \change{approaches} typically decouple node representation learning from the ranking objective and \change{depend heavily on} the 
topological structure of target networks, leading to feature–task inconsistency and \change{poor cross-network generalization}.
This paper proposes a novel framework that leverages causal representation learning to 
\change{obtain} robust \change{and} invariant node embeddings for cross-network ranking tasks. \change{Specifically}, we introduce an influence-aware causal node embedding module within an autoencoder architecture to 
extract node embeddings that are causally related to node importance. \change{Furthermore, we design a unified optimization framework incorporating a causal ranking loss that jointly optimizes reconstruction and ranking objectives, thereby enabling mutual reinforcement between node representation learning and ranking optimization.}
This design allows the 
proposed model to be trained on synthetic networks and to generalize effectively across diverse real-world networks. Extensive experiments on multiple benchmark datasets demonstrate that the proposed model consistently outperforms state-of-the-art baselines in terms of both ranking accuracy and cross-network transferability, offering new insights for network analysis and 
engineering applications—particularly in scenarios where the target network’s structure is inaccessible in advance due to privacy or security constraints.
\end{abstract}

}

\maketitle

\IEEEraisesectionheading{\section{Introduction}\label{sec:introduction}}

Complex networks provide a powerful framework for modeling and analyzing a wide range of systems across diverse domains, including social networks, 
transportation systems, and biological networks~\cite{costa2011analyzing}. In these networks, \change{nodes represent entities within real systems} such as individuals, 
infrastructure components, or functional units, 
while edges capture \change{the} interactions or relationships between them. A key challenge in network science and engineering is identifying important nodes, as they play pivotal roles in 
maintaining network functionality, performance, stability, and robustness\change{~\cite{kaili2025finding, zhan2025measuring}}.
For example, in social networks, \change{important nodes} facilitate targeted information dissemination and optimized resource allocation~\cite{zhang2022influence}.
In transportation networks, they enable vulnerability analysis and robustness enhancement~\cite{zhou2022analyzing}.
In biological networks, they help identify essential genes or proteins for therapeutic interventions~\cite{de2023more}. 
With the growing scale and complexity of real-world networks, accurately identifying critical nodes has become an increasingly important yet 
challenging task~\cite{ma2025node}. Hence, developing accurate and generalizable methods for node importance ranking holds substantial theoretical and practical 
significance~\cite{liu2023identify}.

There has been significant research on node importance ranking in complex networks. 
Traditional topology-based approaches primarily rely on centrality-based measures, such as degree centrality~\cite{bavelas1948mathematical}, eigenvector centrality~\cite{bonacich1972factoring}, and betweenness centrality~\cite{xu2025adaptive}. 
While some centrality measures are straightforward to compute (\emph{e.g.,} degree centrality), others (\emph{e.g.,} betweenness \change{centrality}) can be computationally demanding, particularly in large-scale networks.  Moreover, these methods 
focus on computing individual importance scores rather than modeling relative importance 
\change{relationships} among nodes, which limits their effectiveness \change{for} ranking tasks. In addition, 
most measures capture node importance along a single structural dimension and are often tailored to specific research \change{objectives}, which limits their generalizability across diverse network topologies. For instance, degree-centrality-based \change{methods} can identify high-degree hubs as the most influential nodes in a network. However, such an assumption overlooks the fact that influence can \change{also} arise from nodes occupying structurally balanced or strategically positioned roles, even when their degrees are moderate.



In recent years, deep representation learning–based methods have become a powerful paradigm for node ranking. 
By automatically learning expressive node embeddings, these methods enable downstream tasks such as 
regression to be efficiently performed on the learned embeddings. This category includes 
Graph Convolutional Networks (GCNs) for local neighborhood aggregation~\cite{zhao2020infgcn}, graph embedding techniques for dimensionality reduction~\cite{zhu2021community}, and more sophisticated architectures such as 
{graph attention networks}~\cite{park2019estimating} and graph contrastive learning frameworks~\cite{liu2023node} that capture nuanced structural dependencies. 

Despite their strong representational power, most existing deep representation learning methods focus solely on modeling network \change{structure}—capturing low- and high-order proximities \change{among} nodes—while overlooking node importance information that is crucial for ranking tasks.
For example,  AGNN~\cite{xiong2024vital} 
and CGNN~\cite{zhang2022new}  typically follow a two-stage paradigm, where node representations are first 
learned independently using {graph neural network} techniques, followed by a prediction module for ranking or regression. As task-specific objectives are not 
integrated into the representation learning stage, the resulting embeddings often fail to capture task-relevant information, leading to suboptimal performance in node importance ranking.
Moreover, these deep representation learning–based approaches typically rely heavily on the topology of the target network. 
However, in many real-world scenarios, privacy or security constraints make the network effectively a black box to users, restricting access to 
its underlying structure in advance~\cite{li2021manipulating}. 
For instance, in various strategic or security-sensitive contexts, identifying critical nodes within a non-cooperative network is essential for understanding or mitigating its functionality, yet the network topology is typically unknown in advance due to intelligence and security limitations.
Such situations exemplify a generic problem in network science and engineering—how to identify and rank critical nodes when the topology of the target network is unobservable.  
In this work, we specifically 
investigate whether feature embeddings that capture node importance can be learned through a representation learning model trained exclusively on 
synthetic networks and then transferred across multiple diverse real-world networks.

To achieve this goal, our approach is inspired \change{by} causal representation 
learning~\cite{chu2023continual}, which provides domain-invariant 
representations that generalize across different environments~\cite{zhou2025learning,zhou2022causal}. 
Specifically, we pursue two key objectives: (i) to obtain \change{an} effective node representation suitable for ranking without using the target network structure—enabling applicability in practical scenarios; and (ii) to improve ranking performance through mutual reinforcement 
between node representation learning and ranking optimization—ensuring more reliable and interpretable ranking results. 
Achieving these goals is nontrivial, as it raises two major challenges: 
(i) {how to design an effective unsupervised representation learning strategy that produces 
network-invariant node embeddings,} and (ii) how to ensure that the learned representations are inherently relevant 
and informative for node ranking.



To address the first challenge, we design an influence-aware causal node embedding
module within \change{an} autoencoder to 
capture robust low-dimensional embeddings that exhibit causal relevance to node importance. This mechanism 
enables the learned representations to be network-invariant and to 
generalize effectively to unseen target graphs. To tackle the second challenge, we formulate a unified objective function that jointly optimizes our proposed causal reconstruction and causal ranking losses, enabling a synergistic interaction between representation learning and ranking optimization. 
Our contributions can be summarized as follows:
\begin{itemize}
	\item \textit{Influence-aware causal node embedding mechanism:  We design an influence-aware causal node embedding mechanism for node ranking prediction and integrate it into an autoencoder framework.} Specifically, we construct a node influence variable based on the information propagation process within the training network, and learn the causal relationships between node embeddings and this influence variable. This design encourages the learned node representations to capture network-invariant causal signals related to node importance, 
    enabling the model to be trained solely on synthetic networks and to generalize effectively to diverse real-world graphs.
	\item \textit{Feature–task co-optimization mechanism:} We propose a causal reconstruction loss and a causal ranking loss, and integrate 
    them with a regularization term into a unified objective. This feature-task co-optimization framework jointly 
    optimizes node representation learning and ranking prediction, ensuring that the 
    {resulting} embeddings are directly aligned with the downstream ranking task.
    \item \textit{Extensive empirical studies on real-world networks:}  We conduct \change{comprehensive} experiments to validate the effectiveness of the proposed method. The results demonstrate that it improves the performance of 
    the node ranking model on various real-world networks in terms of both accuracy and generalization. Ablation studies further highlight the contribution of the designed influence-aware causal mechanism.
\end{itemize}

The remainder of this paper is organized as follows.  Some related work is outlined in Section~\ref{sec4}. Problem formulation is introduced in Section~\ref{sec2}. The proposed influence-aware causal autoencoder model for node importance ranking is presented in detail in Section~\ref{sec3}. The experimental results are reported in {Section~\ref{exp}}. Conclusions are drawn in the final section.


 

\section{Related work}\label{sec4}

\begin{table*}[htbp]
	\caption{Comparative analysis of different node importance ranking methods.} 
	\label{tab:related_work} 
	\centering
	\begin{tabular}{m{0.17\textwidth} m{0.25\textwidth} m{0.25\textwidth} m{0.23\textwidth}}
		\toprule
		 Categories & Methods & Contributions &Limitations  \\ 
		\midrule
		
		  Neighborhood-based & DC~\cite{nieminen1974centrality}, $H$-index~\cite{hirsch2005index}, $K$-shell~\cite{kitsak2010identification} & Focus on local structural information. &\multirow{9}{=}{These methods measure node importance from a single structural dimension, which limits their applicability to more complex network structures.} \\
          \cmidrule{2-3}
		  Eigenvector-based & EC~\cite{bonacich1987power}, \change{KC}~\cite{katz1953new}, PageRank~\cite{page1998pagerank} & Leverage spectral characteristics and iterative propagation mechanisms to evaluate global node influence.   \\
          \cmidrule{2-3}
		  Path-based & BC~\cite{freeman1977set}, ECC~\cite{hage1995eccentricity}, CC~\cite{freeman1978centrality}, ViralRank~\cite{iannelli2018influencers} & Assess node importance via paths, distances, or random walks to capture multi-scale topological information.   \\
          \cmidrule{2-4}
          Fusion-based & ISCL~\cite{mohammad2024convex}, HCT~\cite{wang2024identification}, \change{LHGC~\cite{xie2023vital}, EffG~\cite{shang2021identifying}}, CBGN~\cite{liu2023influence}, ITFRSGC~\cite{an2025identifying}, \change{HDF~\cite{zhang2025locating}, GMM~\cite{liu2020gmm}}
          & Blend multiple centrality metrics or topological characteristics to create hybrid measures.
          & These methods have high computational complexity and are sensitive to parameter tuning, demanding manual adjustment.
          \\
          \midrule
          Machine learning-based & ML~\cite{hajarathaiah2024node}, LS-SVM~\cite{wen2018fast}, MFIM~\cite{wang2024multi} & Identify influential nodes by formulating the task as a regression or classification problem. & These methods rely on explicit network topology information, limiting their generalization performance across networks.\\
          \midrule 
		   \multirow{5}{*}{Deep learning-based}  & InfGCN~\cite{zhao2020infgcn},  NDM~\cite{ahmad2024neural}, CGNN~\cite{zhang2022new}, ReGCN~\cite{wu2024graph}, \change{GCNT}~\cite{tang2024gcnt} & Integrate 
           handcrafted features with deep learning architectures for influential node identification. & These methods rely on handcrafted features and are unable to capture complex structural patterns.  \\
           \cmidrule{2-4}
          & AGNN~\cite{xiong2024vital},
          GCRA~\cite{kaili2025finding}, MVPL~\cite{ma2025node},
          RCNN~\cite{yu2020identifying},
          STGCN~\cite{huang2025sparse}, \change{SS-GCN}~\cite{liu2022learning}
          
         & Directly generate latent node embeddings from graph topology and subsequently calculate importance scores. & These methods focus solely on modeling network structures but fail to capture intrinsic features related to node importance, resulting in weak generalization across networks.         
         \\ 
		\bottomrule
	\end{tabular}
\end{table*}

Our framework is a node-ranking method built on causal representation learning. Below, we review the two closely related research areas: 
node importance ranking and causal representation learning.
\subsection{Node importance ranking}
Previous works for identifying vital nodes in complex networks can be categorized into three main groups: traditional topology-based methods, machine learning-based methods and deep learning-based methods, as detailed in Table~\ref{tab:related_work}.
The first category, topology-based node importance estimation methods, mainly includes neighborhood-based, eigenvector-based, path-based and fusion-based approaches. Neighborhood-based methods, such as Degree Centrality (DC)~\cite{nieminen1974centrality}, $H$-index~\cite{hirsch2005index}, and $K$-shell decomposition~\cite{kitsak2010identification}, assess node importance using local structural information. In contrast, eigenvector-based methods like Eigenvector Centrality (EC)~\cite{bonacich1987power}, \change{Katz Centrality (KC)}~\cite{katz1953new} and PageRank~\cite{page1998pagerank} capture global influence propagation. Path-based methods, including Betweenness Centrality (BC)~\cite{freeman1977set}, Eccentricity Centrality (ECC)~\cite{hage1995eccentricity}, Closeness Centrality (CC)~\cite{freeman1978centrality} and ViralRank~\cite{iannelli2018influencers}, evaluate nodes based on their positions and roles in network paths. 
While these methods offer computational simplicity and ease of interpretation, it is evident that these centrality measures depend solely on the network structure and capture only one specific structural aspect. Consequently, they still face challenges in scalability and adaptability when applied to more complex network structures. Fusion-based methods blend multiple centrality metrics or topological characteristics 
to create hybrid measures. Specifically, Mohammad \emph{et al.}~\cite{mohammad2024convex} proposed a novel centrality measure called ISCL (Isolating Clustering Centrality), which employs a convex combination of isolating centrality and clustering coefficient with a tuning parameter to identify influential nodes in large-scale networks efficiently. Wang \emph{et al.}~\cite{wang2024identification} introduced the HCT method, which combines a hypergraph fuzzy gravity model, node centrality distribution characteristics, and TOPSIS to progressively identify important nodes in multi-layer hypergraphs from local to global levels. \change{Xie \emph{et al.}~\cite{xie2023vital} introduced LHGC, which incorporates node degree and higher-order distance in a gravitational model. Shang \emph{et al.}~\cite{shang2021identifying} proposed an effective distance gravity (EffG) model that integrates degree centrality and effective distance to identify influential nodes by considering both static topological information and dynamic interaction patterns in complex networks.} Liu \emph{et al.}~\cite{liu2023influence} proposed a Community-based Backward Generating Network (CBGN) framework for identifying influential nodes by integrating community detection, graph traversal, and an improved reverse-generation strategy to efficiently select high-impact seeds while minimizing influence overlap. An \emph{et al.}~\cite{an2025identifying} proposed the Intrinsic and Topological Features-based Relationship Strength Gravity Centrality (ITFRSGC) that identifies key nodes in unweighted multi-layer networks by modeling hidden relationship strengths through multi-feature fusion and incorporating a dynamic influence range into a gravity centrality model. 
\change{Zhang \emph{et al.}~\cite{zhang2025locating} proposed Higher-order Distance-based Fuzzy centrality methods (HDF) that leverage fuzzy sets and Shannon entropy to identify influential nodes in hypergraphs by considering the collective influence of nodes within a certain higher-order neighborhood. Liu \emph{et al.}~\cite{liu2020gmm} introduced a generalized mechanics model (GMM) that identifies node importance by integrating local degree information and global eigenvector-based weights within a weighted gravity framework, incorporating a truncation radius for scalability.}
Nevertheless, these methods suffer from high computational complexity and sensitivity to parameter tuning, often requiring manual adjustment.
 
Machine learning–based approaches can partially overcome these issues by transforming the task into a regression or classification problem, using models like Linear Regression (LR), 
Multilayer Perceptron (MLP) and Support Vector Machine (SVM). Hajarathaiah \emph{et al.}~\cite{hajarathaiah2024node} proposed a machine learning framework that combines traditional and local relative-change centrality metrics with epidemic simulation outcomes to estimate node importance. Wen \emph{et al.}~\cite{wen2018fast} developed a two-level node ranking framework that integrates AHP-based evaluation 
with Least Squares Support Vector Machine (LS-SVM) learning to approximate global centrality measures using simple local indicators, enabling efficient and accurate node importance estimation. Wang \emph{et al.}~\cite{wang2024multi} introduced a Multi-Factor Information Matrix (MFIM) that integrates node self-influence, neighbor influence, and mutual interaction through hybrid weighting to identify super-propagators in directed weighted social networks. While these methods significantly enhance computational efficiency and provide robust solutions, they still rely on explicit network topology information, which may 
limit their generalization performance in more complex scenarios.

Recent advancements in deep learning, particularly Graph Neural Networks (GNNs), have demonstrated superior performance by learning rich and informative representations of nodes. 
Zhao \emph{et al.}~\cite{zhao2020infgcn} proposed the InfGCN
algorithm that takes neighbor graphs and classic structural features as input into a graph convolutional network for learning node representations, 
followed by a classification layer. 
Zhang \emph{et al.}~\cite{zhang2022new} combined convolutional neural networks (CNN) with GNNs in their CGNN algorithm, where nodes are labeled via the SIR model, and a loss function is optimized to accurately identify critical nodes.
Ahmad and Wang~\cite{ahmad2024neural} proposed a Neural Diffusion Model (NDM) that leverages recurrent neural networks (RNNs) for dynamically modeling temporal influence propagation across networks.
Wu \emph{et al.}~\cite{wu2024graph} proposed ReGCN, a graph convolutional network model that leverages regular equivalence-based similarity to construct node features and employs an adaptive eigenvector selection strategy for influential node identification. \change{Tang \emph{et al.}~\cite{tang2024gcnt} proposed GCNT that integrates graph convolutional networks with graph transformers and a Greedy-LIE label generation strategy to effectively identify influential seed nodes by jointly modeling local structure, global dependencies, and overlapping influence in social networks.} Although these methods have shown considerable impact, they are inherently limited by their reliance on handcrafted features, which often 
reduces their ability to generalize across diverse network structures and capture complex structural patterns. 

To overcome these limitations, Kaili \emph{et al.}~\cite{kaili2025finding} proposed GCRA, a deep learning-based framework that integrates Adaptive Graph Contrastive Learning (AdaGCL) with Deep Reinforcement Learning (DRL) to identify critical nodes by adaptively augmenting graph views and utilizing a dual-optimization reward function for network disintegration. Ma \emph{et al.}~\cite{ma2025node} proposed a multi-view graph prompting framework (MVPL) that integrates node-level soft prompts and graph-level structural prompts to bridge the gap between pre-training and downstream node importance estimation tasks. Yu \emph{et al.}~\cite{yu2020identifying} developed RCNN, a method that 
integrates GCNs with CNN-based adjacency features for efficient critical node detection. \change{Liu \emph{et al.}~\cite{liu2022learning} proposed a self-supervised graph convolutional network (SS-GCN) that models complex dependencies among multivariate time series by integrating graph structure learning with deep neural networks, enabling more accurate detection of anomalous patterns through relational reasoning across variables.}
Xiong \emph{et al.}~\cite{xiong2024vital} introduced the 
AGNN algorithm, which integrates autoencoders with GNNs to generate node embeddings using graph convolutional networks and optimizes ranking predictions based on the 
learned representations. Building on this work, Huang \emph{et al.}~\cite{huang2025sparse} proposed a Sparse Transformer-enhanced Graph Convolutional Network (STGCN) 
to overcome AGNN’s limitations in large-scale and heterogeneous networks. STGCN integrates local topological modeling with global structural dependencies via stochastic anchor attention and masked 
mechanisms, enabling robust node-importance ranking in complex networks. Although these approaches bypass manual feature engineering by employing GNNs to learn topological embeddings and 
are capable of being pretrained on synthetic networks for ranking prediction on other real-world networks, they still face significant limitations. 
The learned representations are often produced by black-box models that lack an explicit and 
intrinsic connection to node importance, thereby limiting their generalizability across diverse networks. Our framework follows a similar paradigm to AGNN~\cite{xiong2024vital} and STGCN~\cite{huang2025sparse}, in which 
models are trained on generative networks to enable ranking prediction on real-world networks. However, 
unlike these approaches, ICAN introduces a causal node embedding mechanism that can capture the causal relevance of representations to node importance, thereby enhancing both interpretability and generalization capability.

\subsection{\change{Task-oriented graph representation learning}}
\change{To bridge the gap between representation learning and downstream applications, recent studies have shifted towards optimizing embeddings specifically for the target task. For instance, Zhu \emph{et al.}~\cite{zhu2021community} proposed a Structural Equivalence-based Non-negative Matrix Factorization (SENMF) method for community detection. By integrating node similarity and structural equivalence into a unified optimization framework involving modularity maximization, SENMF learns embeddings that are intrinsically aligned with the community structure. Similarly, Huang \emph{et al.}~\cite{huang2021representation} introduced a Relational Graph Transformer Network (RGTN) specifically for node importance estimation in knowledge graphs. RGTN employs a learning-to-rank loss to guide the representation learning process, ensuring that the learned embeddings capture relative ranking information and semantic dependencies tailored to the importance estimation task. Wilder \emph{et al.}~\cite{wilder2019end} proposed CLUSTERNET, a decision-focused learning framework that integrates graph neural networks with a differentiable k-means clustering layer to jointly learn graph representations and solve downstream combinatorial optimization problems end-to-end. Deng \emph{et al.}~\cite{deng2021graph} proposed Graph Deviation Network (GDN), an attention-based graph neural network that jointly learns inter-sensor dependency structures and forecasts multivariate time series to detect and explain anomalies as deviations from learned sensor relationships. 
}

\change{While these task-specific embedding methods effectively align feature learning with the target objectives and outperform decoupled two-stage approaches, they primarily rely on fitting statistical correlations between topological features and task labels within the training distribution. Consequently, they may capture spurious correlations rather than the underlying generative mechanisms of node importance. To circumvent this inherent limitation, ICAN fundamentally shifts the learning paradigm by incorporating causal inference. ICAN is designed to identify and learn the invariant structural and relational factors that causally determine node importance, thereby achieving superior robustness and cross-network generalization capacity.}

\subsection{Causal representation learning}
Causal relationships have the ability to capture the fundamental mechanism of data generation and are stable across various contexts, \change{mitigating the impact of confounding factors}~\cite{jiang2025integrating,zhou2025vulnerability}. Learning causal representations is significant for enhancing the robustness of predictive models. Zheng \emph{et al.}~\cite{zheng2018dags} introduced NOTEARS, a novel method that reformulates the DAG discovery problem as a continuous optimization task incorporating acyclicity constraints, enabling solution via numerical techniques. Subsequently, Ng \emph{et al.}~\cite{ng2019graph} expanded upon NOTEARS within a graph autoencoder framework, leveraging multilayer perceptrons (MLPs) to capture and model nonlinear structural equations, thereby advancing the capability to handle complex, non-linear relationships in causal inference. Yang~\emph{et al.} \cite{yang2021learning} integrated a deep autoencoder and a causal structure learning model to learn causal representations using data. Yu \emph{et al.}~\cite{yu2019dag} proposed DAG-GNN, a deep generative model that uses a graph neural network-based variational autoencoder to learn DAG structures from data, effectively generalizing linear structural equation models to capture nonlinear relationships and handle diverse variable types. \change{Wu \emph{et al.}~\cite{wu2022dir}  proposed DIR that constructs intrinsically interpretable graph neural networks by conducting interventions on the training distribution to create interventional distributions, thereby identifying causal rationales invariant across distributions while filtering out spurious patterns for improved out-of-distribution generalization and interpretability. Fan \emph{et al.}~\cite{fan2022debiasing} proposed DisC that learns disentangled causal and bias substructures by employing a parameterized edge mask generator to split input graphs into causal and bias subgraphs, training separate GNN modules with causal-aware and bias-aware loss functions, and generating counterfactual unbiased samples to decorrelate the variables for improved generalization. Chen \emph{et al.}~\cite{chen2022learning} proposed CIGA that extracts invariant subgraphs containing causal information about labels to learn causally invariant representations for guaranteed out-of-distribution generalization on graphs under various distribution shifts. Mo \emph{et al.}~\cite{mo2024graph} proposed GCIL, a causal-inspired graph contrastive learning method that uses spectral graph augmentation to intervene on non-causal factors and incorporates invariance and independence objectives to extract causal information for invariant representation learning.} \change{Most of these} causal representation learning methods aim to uncover causal relationships among variables, but they are not tailored for downstream ranking tasks, making it difficult to ensure the learned representations are effective for ranking.

\change{We summarize the key lessons learned from the related work as follows:}
\begin{itemize}
  \item \change{Most existing node importance ranking methods rely heavily on network topology and decouple representation learning from ranking objectives, which limits their generalization across different networks.}
  \item \change{Task-oriented graph representation learning improves feature--task alignment, but remains largely correlation-driven and vulnerable to spurious, network-specific patterns.}
  \item \change{Causal representation learning offers invariant and robust representations, yet has rarely been explored in the context of node importance ranking.}
\end{itemize}

\change{Motivated by these observations, ICAN integrates a causal ranking loss into the representation learning process, enabling the model to learn robust, low-dimensional node embeddings that are generalizable across diverse networks and well suited for node importance ranking tasks.}

\section{Problem formulation}\label{sec2}
In this section, we define the problem of node importance ranking in complex networks and introduce the key assumptions of the proposed model.
\subsection{{Problem statement}}
Given a graph \( G = \{V, E\} \), where $V$ is the set of nodes and $E$ is the set of edges, \change{the} notation  \( v_i \in V \) denotes a node in $V$, and \( e_{ij} \in E \) denotes an edge from node \( v_i \) to \( v_j \). 
The graph can be represented by an \( n \times n \) adjacency matrix \( \bm{A} \), where \( n \) denotes the number of nodes. The entries of the matrix are defined such that \( \bm{A}_{ij} = 1 \) \change{if} \( e_{ij} \in E \), and \( \bm{A}_{ij} = 0 \) otherwise. 
Let $\bm{X} \in \mathbb{R}^{n \times d}$ \change{denote} the feature matrix. For graphs without node attributes, $\bm{X}$ can be obtained by applying a node embedding method (\emph{e.g.,} \change{node2vec}~\cite{grover2016node2vec}) to the adjacency matrix $\bm{A}$. 



The node importance ranking is a critical problem in network analysis. Given the graph \( G = \{V, E\} \), our task is to assign a score $s_i$ to each node \( v_i \in V \), forming a score vector \(\bm{S} = (s_1, s_2, \dots, s_n)\),  These scores are then used to derive a node importance ranking \(\bm{R} = (r_1, r_2, \dots, r_n)\), where $r_i$ denotes the rank of node $v_i$. The ranking $\bm{R}$ satisfies the following: \( r_i = r_j \) if \( s_i = s_j \), and \( r_i < r_j \) if \( s_i > s_j \), meaning that a higher score corresponds to a higher (\emph{i.e.,} numerically lower) rank. Unlike methods that emphasize the absolute magnitude of scores, our approach focuses on preserving the relative 
ordering of nodes consistent with their underlying importance, making ranking consistency the primary objective.


\subsection{{Assumptions}}

Our proposed method is designed to learn a robust low-dimensional node representation for predicting node 
importance ranking in complex networks. It is worth noting that our approach is proposed based on the following fundamental assumptions. 


\begin{assumption}\label{assumptions1}
	\change{(Markov condition \cite{fields1990independence})  Given a Directed Acyclic Graph (DAG) $\mathbb{G}$, each variable $X$ in $\mathbb{G}$ is independent of any subset of its non-descendants conditioning on its parents.}
\end{assumption}

\begin{definition}
\change{The triplet $<U,\mathbb{G}, P>$, where $U$ denotes the set of variables, $\mathbb{G}$ is a DAG on $U$, and $P$ denotes the probability distribution of $U$, 
is called a Bayesian network iff it satisfies the Markov condition.}
\end{definition}

\begin{assumption}\label{assumptions2}
\change{(Causal graph \cite{pearl2014probabilistic}) For Bayesian Network 
	$<U,\mathbb{G}, P>$, it can be used to express the causal relationships between the variables in $U$. In DAG $\mathbb{G}$, for a pair of directly connected parent-child variables, the parent variable is the direct cause of the child variable, and the child variable is the direct outcome of the parent variable. We assume that the causality between variables can be expressed by $\mathbb{G}$ as a causal graph.}
\end{assumption}


\begin{definition}
\change{($d$-separation~\cite{kuang2022stable,pearl2009causality}) In a DAG $\mathbb{G}$, a path $\pi$ is said to be blocked by a set of nodes $C$ if and only if (i) $\pi$ contains a chain $V_i \rightarrow V_k \rightarrow V_j$ or a fork $V_i \leftarrow V_k \rightarrow V_j$ such that the middle node $V_k$ is in $C$, or (ii) $\pi$ contains a collider $V_i \rightarrow V_k \leftarrow V_j$ such that the middle node $V_k$ is not in $C$ and such that no descendant of $V_k$ is in $C$.}
\change{Furthermore, in a DAG $\mathbb{G}$, we say that two disjoint subsets of vertices $A$ and $B$ are \textbf{d-separated} by a third (also disjoint) subset $C$ if every path between nodes in $A$ and $B$ is blocked by $C$. We then write
$ A \perp_{\mathbb{G}} B \, | \, C. $}

\end{definition}


\begin{assumption}\label{assumptions3}
	\change{(Faithfulness \cite{pearl2009causality,guo2023adaptive}) Consider a distribution $P$ and a DAG $\mathbb{G}$. $P$ is faithful to the DAG $\mathbb{G}$ if
    \[ A \perp B \, | \, C \implies A \perp_{\mathbb{G}} B \, | \, C \]
    for all disjoint vertex sets $A, B, C$.}
\end{assumption}

\begin{definition}
\change{(Markov Blanket (MB)) Consider a DAG $\mathbb{G}$ and a target node $Y$. Then, the Markov Blanket MB of $Y$ includes its parents, its children, and the parents of its children.}
\end{definition}
\change{Faithfulness ensures that all conditional independencies in the data \change{correspond to $d$-separations} 
in the causal graph. \change{Under this assumption and Markov condition}, all other features that are not in the 
Markov blanket set of the target variable (consisting of the parents, children, and spouses) 
are conditionally independent of the target variable, conditioning on the MB set. This means that the information of the non-Markov blanket feature representations is blocked given the Markov blanket feature representations~\cite{yang2021learning}.}

\section{Proposed method}\label{sec3}
In this section, we first show the overall framework of the proposed method, 
the Influence-aware Causal Autoencoder model for Node importance ranking in complex networks (ICAN),  and then describe its components in detail.
\subsection{Overview of the method}
The framework of ICAN is illustrated in Fig.~\ref{framework}. ICAN consists of two core modules: the causal node embedding module and the causal ranking prediction module. The causal node embedding module aims to learn low-dimensional 
node embeddings that are causally related to node importance. Under Assumption \ref{assumptions1}, the causality among node embeddings can be described by a DAG $\mathbb{G}$. Upon obtaining the $p$-dimensional $(p \leq d)$ feature representation, 
the node influence score variable $\bm{Y}$, which can be used to characterize the node importance, is introduced 
as an additional node to form a new hidden layer. The adjacency matrix \( \bm{W} = \{w_{ij}\}_{(p+1)\times(p+1)} \) associated 
with \( \mathbb{G} \), which encodes the causal relationships among  the low-dimensional 
node embeddings \( \bm{H}^{(l)} \), is then integrated into the autoencoder framework to enable message passing under causal mechanisms. 
By optimizing \( \bm{W} \), ICAN can learn optimal representations that are causally linked to node importance. 
The learned latent representations are then fed into the causal ranking prediction module, where the \change{Markov blanket} of the node 
importance score variable is leveraged to predict the ranking. In this way, ICAN can produce robust \change{and} network-invariant node embeddings, enabling cross-network importance ranking predictions. 
In the following, we provide a detailed description of the ICAN framework. 

\begin{figure*}[htbp]
	\centerline{\includegraphics[scale=0.45]{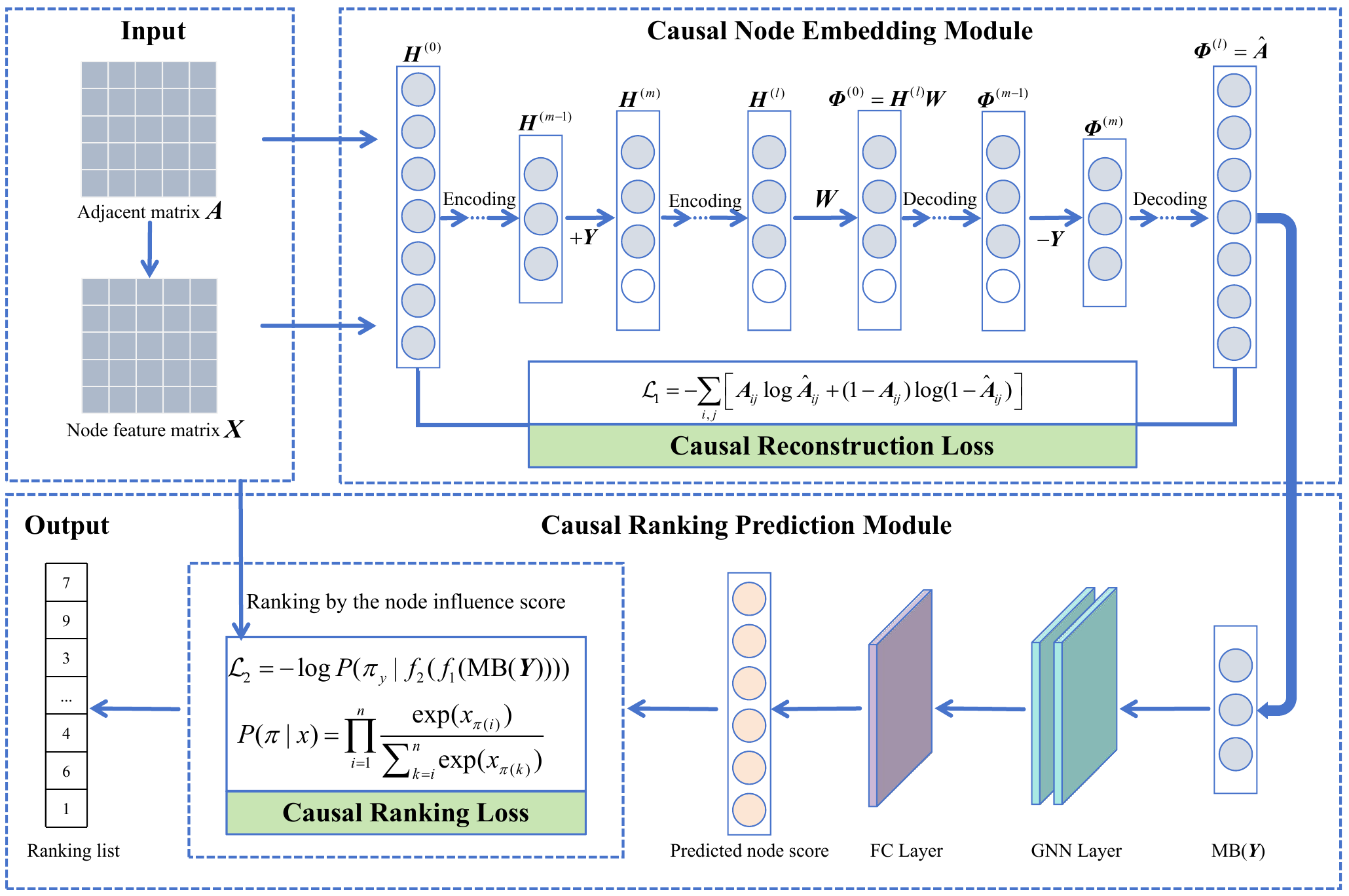}}
	\caption{\change{The framework outline of ICAN. The model takes {the adjacency matrix $\bm{A}$ as inputs} and processes them through two key components:  (1) Causal node embedding module, which learns low-dimensional node embeddings that are causally related to node importance; and (2) Causal ranking prediction module, which utilizes the Markov Blanket of the node 
influence score variable to predict the final ranking list. This model, trained solely on synthetic networks, generalizes effectively to diverse real-world graphs for node importance ranking.  
} }
	\label{framework}
\end{figure*}

\subsection{{Causal node embedding module}}
The goal of the causal node embedding module is to infer causal relationships among low-dimensional node representations. In general, this 
module consists of two main components: an encoder and a decoder.  
We first generate a $d$-dimensional
feature matrix $\bm{X}$ from the graph adjacency matrix $\bm{A}$ using \change{node2vec}. The encoder, typically a Graph Convolutional Network (GCN), then takes both $\bm{X}$ and $\bm{A}$ as input.
It learns to map this input 
into a {$p$}-dimensional low-dimensional representation. The decoder then uses this learned representation to reconstruct the original adjacency matrix. This integrated process of encoding and decoding can be summarized as follows:

\noindent Encoder:
\begin{align}
	& {\bm{H}^{(0)}=\mathrm{ReLU}(\widetilde{\bm{A}}\bm{X}\bm{w}_1^{(0)}+\bm{b}_1^{(0)});}  \\
	&{\bm{H}^{(i)}} = \mathrm{ReLU}(\widetilde{\bm{A}}{\bm{H}^{(i-1)}}\bm{w}_1^{(i)}+\bm{b}_1^{(i)}), i= 1,\cdots,m-1;  \\
	&{\bm{H}^{(m)}} = [{\bm{H}^{(m-1)}}, {\bm{Y}}]; \label{includeys} \\
	&{\bm{H}}^{(i)} = \mathrm{ReLU}(\widetilde{\bm{A}}{\bm{H}^{(i-1)}}\bm{w}_1^{(i)}+\bm{b}_1^{(i)}), i= m+1, \cdots,l.
\end{align}
Decoder:
\begin{align}
	& {\bm{\varPhi}^{(0)}=\bm{H}^{(l)}\bm{W};}  \\
	&{\bm{\varPhi}^{(i)}} = \mathrm{ReLU}({\widetilde{\bm{A}}\bm{\varPhi}^{(i-1)}}\bm{w}_2^{(i)}+\bm{b}_2^{(i)}),i= {1},\cdots,m-1;  \\
	&{\bm{\varPhi}^{(m)}} = {\bm{\varPhi}^{(m-1)}}[:,1:({p}-1)]; \label{removeys}  \\
	&{\bm{\varPhi}^{(l)}}=\mathrm{Sigmoid}(\bm{\varPhi}^{(m)}(\bm{\varPhi}^{(m)})^T).
\end{align}
Here, $\widetilde{\bm{A}}=\bm{D}^{-\frac{1}{2}}(\bm{A}+\bm{I})\bm{D}^{-\frac{1}{2}}$ and $\bm{D}$ is the degree matrix. $l$ denotes the number of hidden layers.
{$M[:,1:({p}-1)]$ denotes the matrix composed of the first $p-1$ columns}. Let
$\bm{\varPhi}^{(l)}=\hat{\bm{A}}$ be the reconstructed adjacency matrix. 
$\bm{w}_1^{(i)}$ and $\bm{b}_1^{(i)}$ denote the weight matrix and bias vector on the $i^{th}$ encoding layer, respectively. 
Similarly, $\bm{w}_2^{(i)}$ and $\bm{b}_2^{(i)}$ denote the weight matrix and bias vector on the $i^{th}$ decoding layer, respectively.
\change{The dimensions of the intermediate matrices are defined as $\bm{H}^{(i)} \in \mathbb{R}^{n \times p}$ and $\bm{\varPhi}^{(i)} \in \mathbb{R}^{n \times (p+1)}$ for $i \in \{0, \dots, m-1\}$; $\bm{H}^{(i)} \in \mathbb{R}^{n \times (p+1)}$ for $i \in \{m, \dots, l\}$; and  $\bm{\varPhi}^{(m)} \in \mathbb{R}^{n \times p}$, $\bm{\varPhi}^{(l)} \in \mathbb{R}^{n \times n}$.} Within this structure, $\bm{H}^{(m-1)} \in \mathbb{R}^{n \times p}$ in the encoding process can be seen as the low-dimensional representations.

\ptitle{Influence-aware causal node embedding mechanism.} Based on Assumption 1, learning the causal relationships among variables can be formulated as 
learning the matrix 
$\bm{W}$ associated with the DAG $\mathbb{G}$. To obtain node embeddings that are causally related to node importance, 
we introduce a pseudo-label variable $\bm{Y}$, termed the node influence score, which characterizes node importance 
based on the SIR model—a classical epidemic framework describing a spreading process where each node can be in one of 
three states: susceptible (S), infected (I), or recovered (R). At each time step, an infected node infects its susceptible neighbors with probability $\gamma$ and recovers with probability $\delta$, after which it cannot be infected again. The process continues until no infected nodes remain.
This model captures the key dynamics of contagion and recovery, which are also relevant to information diffusion and network resilience. Therefore, it provides a reliable basis for evaluating the spreading ability of nodes and for training deep models to understand real-world network dynamics. In our experiments, each node $v$ is initialized as the only infected node, while all others are susceptible. The final number of infected nodes when the process stabilizes is denoted as $F_v$, representing the spreading influence of node $v$. {The average value of $F_v$ over 100 simulations is used as the influence score, denoted by $\bm{Y}_v$. Then we run the SIR model for each node in the network to obtain the influence score variable $\bm{Y}$.


To identify the causal relationship between the node feature variables and  the node influence
score variable, the hidden layer must 
	include the information of $\bm{Y}$  in the training process. Thus, starting with low-dimensional feature representation $\bm{H}^{(m-1)}$, ICAN incorporates $\bm{Y}$ as an additional node to construct a new hidden layer $H^{(m)}$ (Eq.~\eqref{includeys}). 
	Subsequently, $\bm{Y}$ is removed after the $(m-1)^{th}$ decoding layer (Eq.~\eqref{removeys}).
	
The {causal} reconstruction loss $\mathcal{L}_1$ can be defined as 
\begin{equation}\label{l1} 
	\mathcal{L}_{1} = - \sum_{i,j} \left[ \bm{A}_{ij} \log \hat{\bm{A}}_{ij} + (1 - \bm{A}_{ij}) \log (1 - \hat{\bm{A}}_{ij}) \right].
\end{equation}
Based on the above computational process, the reconstructed adjacency matrix $\hat{A}$ can be represented as follows: 
\begin{equation}
	\hat{\bm{A}} = \sigma(\bm{\varPhi}^{(m)} (\bm{\varPhi}^{(m)})^{T}),~~ \bm{\varPhi}^{(m)} = g_2( g_1(\bm{X}, \bm{A})\bm{W}),
\end{equation}  
{where} \change{the adjacency matrix $\bm{W}_{(p+1)\times(p+1)}$ associated 
with \( \mathbb{G} \)} is required to satisfy the constraint {in Eq.~}\eqref{eq:acyclicity_constraint}.  $g_1$ and $g_2$ can be regarded as two nonlinear functions used to describe the encoding and decoding process, respectively. Therefore, it is noted here that the causal node embedding module can be seen as a graph autoencoder.

\change{Utilizing the learned adjacency matrix $\bm{W}$ associated with the DAG $\mathbb{G}$, the Markov blanket of the node influence variable $\bm{Y}$ can be precisely identified. Under the Markov condition (Assumption 1) and Faithfulness assumptions (Assumption 3), $Y$ is conditionally independent of all other variables in $\mathbb{G}$ given its Markov blanket. Generally, $\text{MB}(\bm{Y})$ constitutes the minimal sufficient set that shields $\bm{Y}$ from all other non-descendant nodes.  In practice, while the direct causal dependencies are manifested as non-zero entries in the row and column of $\bm{W}$ associated with $\bm{Y}$, the complete MB set is extracted by traversing the topological connectivity of $\mathbb{G}$. By isolating these causal feature indices, we effectively filter out task-irrelevant embeddings that exhibit only spurious correlations with the target. Consequently, the low-dimensional node embeddings in $\bm{H}^{(m)}$ can be divided into two groups: 
causal node embeddings $\text{MB}(\bm{Y})$, and the other feature representations which may be task-irrelevant features. As such, utilizing this invariant subset for node-level prediction tasks enhances model performance and generalization capability on unseen networks. }

\subsection{{Causal ranking prediction module}} 
Next, we develop the causal ranking prediction module to assign a score to each node in the graph. The causal feature set $\text{MB}(\bm{Y})$ is utilized as the input to the ranking prediction module. The hidden layer of the rank prediction module consists of two layers of GNN and one layer of the FC layer, as expressed by the following equations.
\begin{align}
	& {\bm{E}^{(0)}=\text{MB}(\bm{Y});}  \\
	&{\bm{E}^{(i)}} = \mathrm{Sigmoid}\left(\bm{A}\bm{E}^{(i)}\bm{w}_{3}^{(i)}+\bm{b}_{3}^{(i)}\right), i= 1,\cdots,t;\\
	&{\bm{E}^{(t+1)}} =\bm{E}^{(t)}\bm{w}_{3}^{(t+1)}+\bm{b}_{3}^{(t+1)}.
\end{align}
Here, 
$\bm{A}$ represents the adjacency matrix of the graph, and $\bm{w}_{3}^{(i)}$ and $\bm{b}_{3}^{(i)}$ are the weight matrix and bias matrix respectively, \change{$\bm{E}^{(i)} \in \mathbb{R}^{n \times p}$.}

{\ptitle{Task-aware ranking loss function.}} 
{We introduce a new ranking loss function, CausalListMLE, which can be viewed as an extension of ListMLE~\cite{xia2008listwise} for 
optimizing the overall ranking 
prediction module. This list-wise ranking loss} directly maximizes the likelihood of the entire ranked permutation, rather than focusing 
on individual scores as in point-wise losses like MSE.
By modeling the ranking distribution, CausalListMLE focuses on preserving the correct relative order among all items, 
and typically achieves better performance in ranking-based learning tasks. The causal ranking loss $\mathcal{L}_2$ is defined as
\begin{equation}\label{l2} 
	\mathcal{L}_2=-\log P(\pi_y\mid f_2(f_1(\mathsf{MB}(\bm{Y})))),
\end{equation}
where {the node influence score} $\bm{Y}$ is generated via the SIR model, and $\pi_y$ represents the corresponding ranking results. $f_1$ and $f_2$ denote the GNN layer and FC layer, respectively.
The probability $P(\pi_y \mid f_2(f_1(\mathsf{MB}(\bm{Y}))) )$ follows the Plackett--Luce model, expressed as
\begin{equation}
	P(\pi \mid x) = \prod_{i=1}^{n} 
	\frac{\exp(x_{\pi(i)})}
	{\sum_{k=i}^{n}\exp(x_{\pi(k)})},
\end{equation}
where $x_{\pi(i)}$ is the predicted score of the node at position $i$ in the ranking $\pi$.
By maximizing this likelihood, the causal ranking loss encourages the model to produce \change{a} ranking that closely \change{aligns} with the influence score.

 \change{Remark: It is essential to clarify that while the mathematical formulation of CausalListMLE (Eq.~\eqref{l2}) shares the same Plackett–Luce likelihood structure as standard ListMLE~\cite{xia2008listwise}, our CausalListMLE derives its novelty from a fundamental difference in the input and the resulting optimization objective. Standard ListMLE computes ranking scores $f(\bm{X})$ using all input features, often leading to models that overfit spurious correlations. In contrast, CausalListMLE is fundamentally designed to restrict the input of the scoring function to the causal feature subset $\text{MB}(\bm{Y})$ (the Markov Blanket of the node influence). This restriction enables the model to co-optimize the alignment between the causal embeddings and the ranking order, thereby  
yielding invariant and importance-related features, which are critical for achieving robust cross-network generalization.}



\subsection{{The objective function}}
We have designed the causal reconstruction loss $\mathcal{L}_1$ to ensure the learned representations accurately \change{capture} the causal dependencies underlying node importance, and the causal ranking loss $\mathcal{L}_2$ to optimize them for the downstream ranking task.
In addition, to mitigate model overfitting and prevent performance degradation, we design a regularization term $\mathcal{L}_3$ on the weight parameters to the objective function $\mathcal{L}$:
\begin{align}\label{l3}    
	\mathcal{L}_3 = \sum_{\substack{i=1\\i\not=m}}^l (||\bm{w}_1^{(i)}||^2)  + \sum_{\substack{i=1}}^{m-1} (||\bm{w}_2^{(i)}||^2)+\sum_{\substack{i=1}}^{t+1} (||\bm{w}_3^{(i)}||^2). 
\end{align}

\change{Remark: The term $\mathcal{L}_3$ corresponds to the standard $L^2$ (weight-decay) regularization applied to model weights. This penalty constrains the magnitude of the weight parameters by adding their squared $\ell_2$ norm to the loss, which effectively reduces overfitting by preventing excessively large weights. Since each bias parameter controls only a single variable, not regularizing it does not lead to significant variance. Furthermore, 
regularizing bias parameters can result in noticeable underfitting. Therefore, in neural networks, regularization is typically applied only to weight parameters, not bias parameters \cite{goodfellow2016deep}. Also, in this paper, the regularization term $\mathcal{L}_3$ does not impose constraints on bias terms.}

\ptitle{Feature–task co-optimization mechanism.} {According to the previous defined losses in Eqs.~\eqref{l1}, \eqref{l2} and \eqref{l3}}, the objective function of ICAN can be defined as
\begin{align}\label{objectiveall}          
	&\mathcal{L} = - \frac{\lambda_1}{N}\sum_{i,j} \left[ \bm{A}_{ij} \log \hat{\bm{A}}_{ij} + (1 - \bm{A}_{ij}) \log (1 - \hat{\bm{A}}_{ij}) \right] \nonumber \\&  -\lambda_2\log P(\pi_y \mid f_2(f_1(\text{MB}(\bm{Y})))) \nonumber \\& +\lambda_3\left[\sum_{\substack{i=1\\i\not=m}}^l (||\bm{w}_1^{(i)}||^2)  + \sum_{\substack{i=1}}^{m-1} (||\bm{w}_2^{(i)}||^2)+\sum_{\substack{i=1}}^{t+1} (||\bm{w}_3^{(i)}||^2)\right]. 
\end{align}
{To ensure that the learned $\bm{W}$ forms a DAG, this objective must satisfy the acyclicity constraint~\cite{zheng2018dags}: }
\begin{equation}\label{eq:acyclicity_constraint}
	\mathrm{tr}\!\left(e^{\bm{W}  \odot \bm{W} }\right) - (p + 1) = 0,
\end{equation}
where \( \mathrm{tr}(\cdot) \) denotes the trace operator, \( e^{\bm{W}} \) represents the matrix exponential, and \( \odot \) denotes the Hadamard product.
\begin{lemma}
	Let \( \mathbb{G} = (U, E) \) be a graph with adjacency matrix \( \bm{W}  \in \mathbb{R}^{(p+1)\times(p+1)} \).  
	Then, the \((i,j)\)-th element of \( \bm{W} ^k \), denoted by \( \bm{W} _{ij}^{(k)} \), represents the number of paths of length \( k \) from node \( v_i \) to node \( v_j \), where \( i,j = 1,\ldots,p+1 \).
\end{lemma}
\begin{theorem}
The adjacency matrix \( \change{\bm{W}} \) corresponds to an acyclic graph \( \mathbb{G} \) if and only if Eq.~\eqref{eq:acyclicity_constraint} holds. 
\end{theorem}

\begin{proof}
	Let \( \bm{Q} = \bm{W} \odot \bm{W} \). Clearly, \( \bm{Q} \ge 0 \), \emph{i.e.}, 
	$$\bm{Q} \in \mathbb{R}^{(p+1)\times(p+1)}_+.$$ 
	
    a) (Sufficiency):
	If Eq.~\eqref{eq:acyclicity_constraint} holds, then
	\begin{align}\label{eq:trace_expand}
		\mathrm{tr}(e^{\bm{Q}}) - (p+1)
		&= \mathrm{tr}(\bm{I}) - (p+1)
		+ \mathrm{tr}(\bm{Q})
		+ \frac{1}{2!}\mathrm{tr}(\bm{Q}^2)
		+ \cdots \nonumber \\
		&= \mathrm{tr}(\bm{Q})
		+ \frac{1}{2!}\mathrm{tr}(\bm{Q}^2)
		+ \cdots = 0.
	\end{align}
	Since all entries of \( \bm{Q} \) are nonnegative, it follows that
	\(\mathrm{tr}(\bm{Q}^k) = 0\) for all \( k \ge 1 \), implying \( q_{ii}^{(k)} = 0 \) for all \( i \).  
	By Lemma~1, there are no cycles in \( \mathbb{G} \); hence, \( \mathbb{G} \) is acyclic.
	
    b) (Necessity):
	Conversely, if \( \mathbb{G} \) is acyclic, Lemma~1 implies \( q_{ii}^{(k)} = 0 \) for all \( i \) and \( k \).  
	Thus, \( \mathrm{tr}(\bm{Q}^k) = 0 \), and therefore
	\begin{align}
		{0 = \mathrm{tr}(\bm{Q})
		+ \frac{1}{2!}\mathrm{tr}(\bm{Q}^2)
		+ \cdots = \mathrm{tr}(e^{\bm{Q}}) - (p+1).}
	\end{align}
	Hence, Eq.~\eqref{eq:acyclicity_constraint} holds, proving the necessity.
\end{proof}

To solve above equality-constrained problem, the augmented Lagrangian method~\cite{ng2019graph} is employed to convert it into an unconstrained problem. After introducing the Lagrange multiplier and the penalty term, the objective function in Eq.~\eqref{objectiveall} is reformulated as the following augmented Lagrangian function:
\begin{align}\label{objectiveall1}          
	&\mathcal{L} = - \frac{\lambda_1}{N}\sum_{i,j} \left[ \bm{A}_{ij} \log \hat{\bm{A}}_{ij} + (1 - \bm{A}_{ij}) \log (1 - \hat{\bm{A}}_{ij}) \right] \nonumber \\&  -\lambda_2\log P(\pi_y \mid f_2(f_1(\text{MB}(\bm{Y})))) \nonumber \\& +\lambda_3\left[\sum_{\substack{i=1\\i\not=m}}^l (||\bm{w}_1^{(i)}||^2)  + \sum_{\substack{i=1}}^{m-1} (||\bm{w}_2^{(i)}||^2)+\sum_{\substack{i=1}}^{t+1} (||\bm{w}_3^{(i)}||^2)\right] \nonumber \\&+\alpha h(\bm{W}) + \frac{\rho}{2}|h(\bm{W})|^2],
\end{align}
where $\alpha$ denotes the Lagrange multiplier, $\rho>0$ denotes the penalty parameter and $h(\bm{W})=\text{tr}(e^{\bm{W}\odot \bm{W}}) - (p+1)$. 
We then have the following update rules for the adjacency matrix $\bm{W}$, $\alpha$ and $\rho$:
\begin{align}
	\bm{W}^{(t+1)} &= \arg \min~\mathcal{L}(\bm{W}^{(t)},\alpha^{(t)}, \rho^{(t)}), \label{rules_1}\\
	\alpha^{(t+1)} &= \alpha^{(t)} + \rho^{(t)} h(\bm{W}^{(t+1)}), \label{rules_2}\\ 
	\rho^{(t+1)}&=\begin{cases} \beta \rho^{(t)}, & \mbox{if   } |h(\bm{W}^{(t+1)})| \le \theta |h(\bm{W}^{(t)})|, 
		\\ \rho^{(t)},  & \text{otherwise}, \end{cases}\label{rules_3}
\end{align}
where $t$ denotes the $t_{th}$ iteration; $\beta >1$ and $\theta <1$ are two tuning hyperparameters. The gradient descent method is applied to solve the problem Eq.~\eqref{rules_1}, thereby obtaining the adjacency matrix $\bm{W}$. Following this, the Lagrange multiplier $\alpha$ and the penalty parameter $\rho$ are updated accordingly. The proposed ICAN algorithm is shown in Algorithm~\ref{code}.

\begin{algorithm}[!ht]      
	\renewcommand{\algorithmicrequire}{\textbf{Input:}} 	      
	\renewcommand{\algorithmicensure}{\textbf{Output:}} 	      
	\caption{ICAN} \label{code}        
	\label{power}          
	\begin{algorithmic}[1]                          
		\REQUIRE  Adjacency matrix $\bm{A}$; parameters $\lambda_1$, $\lambda_2$, $\lambda_3$, $l$, $m$, $t$, $d$, $\beta$, $\theta$, $\delta$, $\gamma$, $T$;
		\ENSURE Prediction ranking $\hat{\bm{Y}}$;
		\STATE Initialization:  
		Initial weight and bias parameters randomly; Let $\bm{W} = \mathrm{\bm{I}}, \rho=1, \alpha=0, t=0$.
        \STATE {Acquire node feature matrix $\bm{X}$ using node embedding methods;}
		\WHILE {$t\le T$} 
		\STATE Update the adjacency matrix $\bm{W}^{(t+1)}$ by Eq.~\eqref{rules_1};
		\STATE Update $\alpha^{(t+1)}$ and $\rho^{(t+1)}$ by Eq.~\eqref{rules_2} and Eq.~\eqref{rules_3}, respectively;
		\STATE $t \gets t + 1$;
		\ENDWHILE
		\STATE Extract the causal representation $\text{MB}(\bm{Y})$ by $\bm{W}$ and $\bm{H}^{(m)}$. 
		\STATE Calculate the output of ranking prediction module $\hat{\bm{Y}}$.
		\RETURN $\hat{\bm{Y}}$
	\end{algorithmic}      
\end{algorithm}

\subsection{Computational complexity}
The computational complexity of ICAN primarily stems from its causal node embedding module and causal ranking prediction module. For the autoencoder in the causal node embedding module, the complexity of encoder is $O(|E|d+nd^2)$ and the complexity of decoder is $O(n^2d)$, resulting in an overall complexity for the Autoencoder of $O(|E|d+nd^2+n^2d)$, where $n$ and $d$ denote the number of nodes and feature dimensions of the input $\bm{X}$. $\lvert E \rvert$ is the number of edges. As for the causal ranking prediction module, the complexity of the GNN layer is $O(|E|d+nd^2)$ and the complexity of the FC layer is $O(nd)$. Therefore, the total complexity of the causal ranking prediction module is $O(|E|d+nd^2)$. Combining both components, the overall computational complexity of the proposed model is $O(|E|d+nd^2+n^2d)$.

\section{Experiment}\label{exp}

In this section, experiments are performed on the real-world networks to validate the effectiveness of the proposed ICAN method in the node importance ranking problem.
In the following, the experimental settings such as the used datasets, comparison methods, evaluation criterion and model parameters for the experiments are described first, followed by a detailed analysis of the results.
To systematically evaluate ICAN, we define three research questions (RQs)—"Three CANs Test for ICAN":
\begin{itemize}
	\item RQ1: Can ICAN trained without reliance on the target network learn causally relevant embeddings of node importance and generalize effectively across diverse real-world network structures?
	\item RQ2: Can the causal node embedding mechanism, the feature–task co-optimization mechanism, and the causal ranking loss function in ICAN improve the cross-network ranking performance?
    \item RQ3: Can ICAN achieve comparable performance when trained on different types of generative networks, indicating robustness to the network generation process?
\end{itemize}

\subsection{Description of datasets}

We train the model on five synthetic graphs and evaluate it on \change{eight} real-world networks to validate its effectiveness. 

{\ptitle{Training datasets.}} Five representative synthetic network models are used for training, including Barabási–Albert (BA)~\cite{barabasi2009scale}, Extreme Homogeneous (EH)~\cite{lou2020towards}, Erdös-Rényi (ER) random-graph~\cite{erdHos1961strength}, Q-Snapback (QS)~\cite{lou2018toward} and Random Hexagon (RH)~\cite{lou2018toward}. The BA networks are generated using the preferential attachment mechanism, where new nodes are sequentially added and connected to existing nodes with a probability proportional to their current degrees, leading to a scale-free degree distribution. The EH networks are obtained by performing additional random edge rectifications on ER networks to transform their degree distribution from Poisson to near-uniform. The ER random graphs are generated by connecting each pair of nodes independently with a fixed probability $p$, producing networks with a binomial degree distribution. The QS networks consist of a directed backbone chain with multiple probabilistic "snapback" edges linking newly added nodes to previously added ones within a defined range, thereby enhancing backward connectivity and robustness. The RH networks are composed of randomly connected hexagonal substructures that capture local clustering and spatial organization similar to lattice-like systems. All trained networks comprise 1,000 nodes with an average degree of 4. 

{\ptitle{Target datasets.}} {We use the following real-world networks to test the model trained on the synthetic networks:}
\begin{itemize}
	\item Karate~\cite{zachary1977information}: A human social network, in which each node represents a member in the club and each link shows the relationship between two members.
	\item Jazz~\cite{gleiser2003community}: A human collaboration network, in which each node represents a jazz musician and each link indicates that two musicians have performed together in the same band.
	\item Email-univ~\cite{guimera2003self}: An email communication network from the University Rovira i Virgili in Spain, where nodes represent users and edges indicate at least one email was exchanged.
	\item USAir~\cite{kunegis2013konect}: A directed weighted network represents the flight connections among 1,574 U.S. airports in 2010. Each node corresponds to an airport, and each directed edge indicates the existence of at least one flight from the origin to the destination airport. The edge weight quantifies the total number of flights operated on that route during the year.
	\item Vidal~\cite{rual2005towards}: A network represents an initial version of a proteomescale map of Human binary protein–protein interactions.
	\item Email-dnc~\cite{rossi2015network}: A directed unweighted network derived from the 2016 Democratic National Committee (DNC) email leak. Each node represents an individual user, and each directed edge indicates that at least one email was sent from one user to another.
    \item \change{Figeys~\cite{ewing2007large}: A large-scale human protein–protein interaction network derived from mass spectrometry–based experimental assays, where nodes represent human proteins and edges denote experimentally detected physical interactions between them.}
    \item \change{Oz~\cite{freeman1998exploring}: A friendship network captures interpersonal relationships among 217 residents of a university residence hall, where each node represents an individual and each edge denotes a reported friendship connection between two residents. }
\end{itemize}

In our experimental setup, the model is trained exclusively on synthetic networks and evaluated on real-world networks. This setting requires the model to generalize across graphs with different 
numbers of nodes. To achieve this, we keep the second dimension of the feature matrix $\bm{X}$ (\emph{i.e.,} the node feature dimension) the same for both the training and test graphs. In this way, 
the input dimension of the GCN weight matrix $\bm{w}$ in the node-wise update function
$
\mathrm{ReLU}(\widetilde{\bm{A}}\bm{X}\bm{w} + \bm{b})
$
remains unchanged between training and testing. The fixed dimensionality of $\bm{w}$, together with parameter sharing across all nodes, allows the same set of parameters to be directly applied to graphs of different sizes.
These design choices endow our model with strong inductive generalization capabilities, enabling robust performance in cross-network evaluation scenarios.

\subsection{Comparative methods}
To demonstrate the advantages of the ICAN method in node ranking, we compare it against a range of state-of-the-art baselines widely used for measuring node importance in complex networks. Specifically, the comparison includes five traditional topology-based methods—encompassing neighborhood-based, eigenvector-based, and path-based techniques—along with six recently developed deep learning-based approaches. A brief summary of these methods is provided below.


\begin{itemize}
	\item Degree Centrality (DC)~\cite{bonacich1972factoring}. It quantifies the influence of a node based on the number of direct connections it has to other nodes in the graph. 
    The degree centrality of a node $u$ is formally defined as
	\[
	\mathrm{DC}(u) = \frac{d(u)}{n - 1},
	\]
	where $d(u)$ is the degree of node $u$ (i.e., the number of edges connected to $u$), and $n$ is the total number of nodes in the network.
	
	\item Betweenness Centrality (BC)~\cite{freeman1977set}. It measures the importance of a node by quantifying how frequently it appears on the shortest paths between all pairs of nodes in the network. 
	\[
	\mathrm{BC}(u) = \sum_{s \neq t \neq u} \frac{g_{st}^{(u)}}{g_{st}},
	\]
	where $g_{st}$ denotes the number of shortest paths from node $s$ to node $t$, and $g_{st}^{(u)}$ is the number of those paths that pass through node $u$.
	
	\item Eigenvector Centrality (EC)~\cite{bonacich1987power}. It assigns each node a centrality score based not only on its own connectivity but also on the importance of the nodes it is connected to. 
    The EC score $\bm{X}_u$ of node $u$ satisfies
	\[
	\bm{X}_u = \frac{1}{\lambda} \sum_{k} \bm{A}_{uk} \bm{X}_k,
	\]
	which, in matrix form, becomes
	\[
	\lambda \bm{X} = \bm{A} \bm{X},
	\]
	where $\bm{A}$ is the adjacency matrix of the network, $\bm{X}$ is the eigenvector of centrality scores, and $\lambda$ is the largest eigenvalue of $\bm{A}$. 
	
	\item $H$-index (HI)~\cite{hirsch2005index}. 
    The $H$-index of a node considers the degrees of its neighbors and reflects how many of them are themselves influential. For a node $u$ with neighbors $j_1, j_2, \ldots, j_k$, the $H$-index is defined as
	\[
	\mathrm{HI}(u) = H\left( k_{j_1}, k_{j_2}, \ldots, k_{j_k} \right),
	\]
	where $k_{j_i}$ is the degree of neighbor $j_i$, and $H(\cdot)$ is a function that returns the largest integer $x$ such that at least $x$ neighbors of node $u$ have degrees no less than $x$.
    \item $K$-Shell (KS)~\cite{kitsak2010identification}. The $K$-shell algorithm operates through an iterative pruning process.   
    It iteratively peels away network layers to assign each node a {$k$}-value indicating its coreness based on residual connections.
    A higher $K$-shell value indicates greater node influence.
    \item GNN-Bet~\cite{maurya2021graph}. This method enables node ranking by employing a dual-path aggregation mechanism that separately propagates features along incoming and outgoing shortest paths using row-wise modified adjacency matrices, with final ranking scores generated through multiplicative combination of the aggregated path representations.
    \item GNN-Close~\cite{maurya2021graph}. This approach achieves node ranking through a hierarchical feature aggregation scheme that combines normal adjacency operations in the initial layer with column-wise modified matrices in subsequent layers to constrain feature propagation along shortest paths, producing final ranking scores via summation of layer-wise outputs.
	\item RCNN~\cite{yu2020identifying}. This method leverages convolutional neural networks to extract features from node-centric subgraph structures and ranks node importance based on learned representations.
	\item CGNN~\cite{zhang2022new}. This method enhances RCNN by incorporating additional GNN layers, improving the capture of topological dependencies for more accurate vital node identification.
	\item AGNN~\cite{xiong2024vital}. This method combines a GCN-based autoencoder with a GNN ranking head to jointly learn structural embeddings and predict node importance using a listwise ranking loss.
    \item STGCN~\cite{huang2025sparse} \footnote{Since AGNN and STGCN employ the adjacency matrix as the input to their autoencoder and do not explicitly address the issue of inconsistent numbers of nodes between training and testing networks, 
we instead incorporate the generated node features by the adjacency matrix as the input to the AGNN and STGCN models.}. This method leverages a hybrid architecture combining Sparse Transformer layers for global dependency modeling and GCNs for local topological features to rank key nodes in complex networks.
\end{itemize}

\subsection{Evaluation criterion}
In this section, we introduce the evaluation metric used to measure the effectiveness of the predicted ranking results  
by ICAN, Kendall’s $\tau$ coefficient~\cite{knight1966computer}. This non-parametric statistic is widely employed to 
measure the ordinal association between two ranking lists. It quantifies the similarity between two ranking orders 
by comparing the numbers of concordant and discordant pairs. \change{The definition of Kendall’s $\tau$ is given by}
\[
{\tau = \frac{2(N_c - N_d)}{n(n-1)}},
\]
where $n$ is the number of items (or nodes), $N_c$ is the number of concordant pairs, $N_d$ is the number of discordant pairs, and $\tfrac{1}{2}n(n-1)$ is the total number of possible item pairs. For two rankings $\{(x_i, y_i)\}$ and $\{(x_j, y_j)\}$, a pair $(i,j)$ is considered concordant if $(x_i - x_j)(y_i - y_j) > 0$, discordant if $(x_i - x_j)(y_i - y_j) < 0$, and neither if $x_i = x_j$ or $y_i = y_j$. \change{The value of Kendall’s $\tau$ ranges from $-1$ to $1$}: a value of $1$ indicates complete agreement between rankings, $-1$ indicates complete disagreement, and $0$ indicates no correlation. This metric is particularly useful for evaluating graph-based ranking models, as it reflects how well the predicted node ranking preserves the ground-truth ordering.

\subsection{Implementation details}
The parameter settings for ICAN and the comparative methods are \change{summarized} as follows. 

For the ICAN method, 
set the number of hidden layers to be $l=5$ and each hidden layer to have \change{$p=32$} neurons.  
Let  $$d=128, m = 3, t = 2, \beta =  10, \theta =  0.25.$$ The learning rate is \change{set to} $\mu=0.001$ and the maximum number of iterations is $T=10$. 
In addition, we set the recovery probability $\delta = 1$ and the infection probability $\gamma = 1.5 \times \gamma_c$ to ensure effective spreading within the network. The infection threshold $\gamma_c$ is derived from mean-field theory~\cite{moreno2002epidemic}  as
\begin{equation}
	\gamma_c = \frac{\langle k \rangle}{\langle k^2 \rangle - \langle k \rangle},
\end{equation}
where $k$ denotes the node degree and $\langle \cdot \rangle$ indicates the average over all nodes. 

For other comparison methods, all the parameters are set to the default. For model evaluation, Kendall’s $\tau$ coefficient on the test networks is adopted as the primary evaluation metric. To reduce the impact of randomness on experimental results, both ICAN and the comparative methods are run 5 times each, with average values recorded. All experiments are conducted on a computer running Windows 10, equipped with an Intel(R) i5-10400 CPU at 2.90 GHz and 16 GB of memory.

\subsection{{{Results on various real-world networks (RQ1)}}} \label{RQ1}
In this subsection, we present a detailed analysis of the experimental results across various datasets, highlighting performance outcomes and key insights from each.

\subsubsection{Training on BA-generated network and testing on various real-world networks}

In the first experiment, we conduct experiments by training on the BA generated network and testing on different real-world networks to validate the model's effectiveness. Table~\ref{BA training} presents the predicted Kendall's correlation coefficients across different datasets. The experimental results comprehensively evaluate the proposed ICAN model against several baseline methods for node ranking across \change{eight} real-world networks. \change{As shown in the table}, ICAN demonstrates superior and consistent performance. Notably, ICAN achieves the highest Kendall's coefficient on every individual dataset.
This culminates in the best overall average Kendall's coefficient of 0.7707, securing the top rank among all compared methods. {The fact that ICAN, trained on a BA synthetic network, generalizes effectively to diverse 
real-world networks 
highlights its robust cross-network predictive capability. This can be attributed to the ICAN learning network-invariant 
node embeddings that transfer reliably across different networks.} 

\change{Overall, in terms of average Kendall's coefficient , ICAN ranks first in all datasets. 
From the experimental results above, we can summarize the following conclusions:}
\begin{itemize}
    \item \change{Compared to traditional topology-based methods, ICAN significantly outperforms heuristic baselines such as DC, BC and KS across all datasets. This demonstrates the limitations of methods that rely on single structural dimensions. In contrast, ICAN effectively captures complex, non-linear dependencies and higher-order topological information that traditional centrality measures fail to identify, proving the necessity of deep representation learning for node ranking.}
    \item \change{Compared to deep learning-based methods, ICAN consistently outperforms state-of-the-art baselines, including RCNN, CGNN, AGNN, and STGCN. While frameworks like AGNN and STGCN also utilize autoencoders, they typically decouple representation
learning from the ranking objective and lack explicit causal constraints. The superior performance of ICAN demonstrates that its influence-aware causal node embedding mechanism ensures the learned representations are causally linked to node importance, thereby enhancing robustness and generalization across diverse networks. Furthermore, by jointly minimizing causal reconstruction and ranking losses, ICAN adopts a unified optimization strategy that fosters mutual reinforcement between feature extraction and ranking prediction. This alignment ensures that the network-invariant causal representations are directly tailored for the downstream ranking task, leading to superior cross-network transferability. }
\end{itemize}}

\begin{table*}[htbp]
	\caption{Kendall's coefficient of various methods (trained on BA network) on different real-world test networks. }
	\label{BA training}
	\centering
		\begin{tabular}{ccccccccccc}
			\toprule
				Methods  & Karate & Jazz & Email-univ & USAir &Vidal &Email-dnc &\change{Figeys} &\change{Oz} & Average &Rank \\
			\midrule
			DC &0.6749 &0.7903  &0.6236   &0.6119 &0.4648 &0.5613 &\change{0.5950} &\change{0.7587} &0.6351 &6 \\
			BC &0.5685&0.4643&0.2374&0.4407&0.5876&0.4032 &\change{0.5321} &\change{0.5060} &0.4675   &12  \\
			EC &0.8069&0.8212&0.2534&0.6554&0.2452&0.5503 &\change{0.6546} &\change{0.7684}   &0.5944 &7  \\
			HI &0.6654&0.8271&0.6854&0.6284&0.7568&0.5672 &\change{0.6245}&\change{0.6299}  &0.6946   &4   \\
            KS &0.6421 &0.7713&0.6828&0.6369&0.7619&0.4988 &\change{0.6299}&\change{0.7069} &0.6663   &5 \\
            GNN-Bet&0.6530&0.6789&0.3645&0.5350&0.2866&0.3662 &\change{0.5620} &\change{0.4902} &0.4921  &10\\
           GNN-Close&0.5352&0.6551&0.3629&0.5818&0.2429&0.4442 &\change{0.5453} &\change{0.3832} &0.4688  &11\\
			RCNN &0.7528&0.8455&0.7522&0.6709&0.9035&0.5432 &\change{0.6025} &\change{0.7914} &0.7328  &2  \\
			CGNN &0.7058&0.7911&0.6488&0.5835&0.9442&0.5649 &\change{0.5959}&\change{0.7322} &0.6958  &3 \\
			AGNN &0.6505&0.7584&0.3100&0.6429&0.4108&0.4946 &\change{0.5989}&\change{0.2744} &0.5176  &9  \\
            STGCN &0.5640&0.7000&0.7327&0.3972&0.5784&0.4248 &\change{0.5975}&\change{0.6504} &0.5806 &8  \\
			ICAN  &\underline{0.8090} &\underline{0.8504} &\underline{0.7625}&\underline{0.6758}&\underline{0.9524}&\underline{0.5740}   &\change{\underline{0.6708}}&\change{\underline{0.8051}}&\underline{0.7625}&\underline{1} \\	
			\bottomrule
	\end{tabular}
\end{table*}

\subsubsection{Training on ER-generated network and testing on various real-world networks} 
Next, we conduct experiments by training on the ER generated network and testing on different real-world networks. Table~\ref{ER training} presents the Kendall’s $\tau$ coefficients of various methods evaluated across \change{eight} real-world networks. As shown, ICAN attains the highest Kendall’s $\tau$ values on \change{four} out of \change{eight} datasets, with an average coefficient of 0.7732, ranking first overall. This clearly demonstrates its strong generalization ability when trained on the ER-generated network and tested on structurally diverse real-world networks.

Nevertheless, ICAN exhibits relatively lower performance on certain networks, such as USAir and Email-dnc. This performance gap may be largely attributed to the distributional discrepancy between ER training network and the target real-world networks, particularly in terms of degree heterogeneity and community structure. The ER network is characterized by a nearly uniform degree distribution, while many real-world networks display heavy-tailed or scale-free properties, resulting in markedly different local connectivity patterns. Consequently, the model—though effective in capturing causal dependencies among node features—may face challenges in fully adapting to unseen structural regimes. This also motivates our future work to investigate strategies for selecting or generating training networks that better align with the topological properties of specific application domains.


\begin{table*}[htbp]
	\caption{Kendall's coefficient of various methods (trained on ER network) on different real-world test networks.}
	\label{ER training}
	\centering
		\begin{tabular}{ccccccccccc}
			\toprule
		Methods  & Karate & Jazz & Email-univ & USAir &Vidal &Email-dnc &\change{Figeys} &\change{Oz} & Average &Rank \\
		\midrule
		DC &0.6749 &0.7903  &0.6236   &0.6119 &0.4648 &0.5613 &\change{0.5950}&\change{0.7587} &0.6351&6 \\
		BC &0.5685&0.4643&0.2374&0.4407&0.5876&0.4032 &\change{0.5321}&\change{0.5060}&0.4675&11  \\
		EC &0.8069&0.8212&0.2534&0.6554&0.2452&0.5503 &\change{\underline{0.6546}}&\change{0.7684} &0.5944&7  \\
		HI &0.6654&0.8271&0.6854&0.6284&0.7568&\underline{0.5672} &\change{0.6245}&\change{\underline{0.8017}} &0.6946&3   \\
        KS &0.6421 &0.7713&0.6828&0.6369&0.7619&0.4988 &\change{0.6299}&\change{0.7069} &0.6663& 5\\

        GNN-Bet  &0.6137&0.6394&0.3744&0.5565&0.4125&0.4566 &\change{0.5365}&\change{0.5691} &0.5198&10\\
        GNN-Close&0.6887&0.2779&0.4550&0.5470&0.4079&0.4122 &\change{0.5343}&\change{0.1673} &0.4363&12\\
		RCNN &0.7094&0.8014&0.8141&0.6823&0.9119&0.5490 &\change{0.5576}&\change{0.7507} &0.7221&2 \\
		CGNN &0.6749&0.6214&0.6323&0.5835&0.9442&0.5649 &\change{0.5959}&\change{0.7322} &0.6687&4 \\
		AGNN &0.7622&0.4889&0.2950&\underline{0.6829}&0.4008&0.4534 &\change{0.5091}&\change{0.6549} &0.5309&9  \\
        STGCN &0.5568 &0.6525&0.7141&0.3952&0.5792&0.3735 &\change{0.3833}&\change{0.6751} &0.5412 &8\\
		ICAN  &\underline{0.8092} &\underline{0.8460} &\underline{0.8228}&0.6545&\underline{0.9524}&0.5540 &\change{0.6312}&\change{0.7820}&\underline{0.7565}&\underline{1} \\
			\bottomrule
	\end{tabular}
\end{table*}

\subsection{Ablation study (RQ2)}

\begin{figure*}[htbp]
	\centering
	\subfloat[
    \change{Ablation for causal mechanism.}]{\includegraphics[width=.32\textwidth]{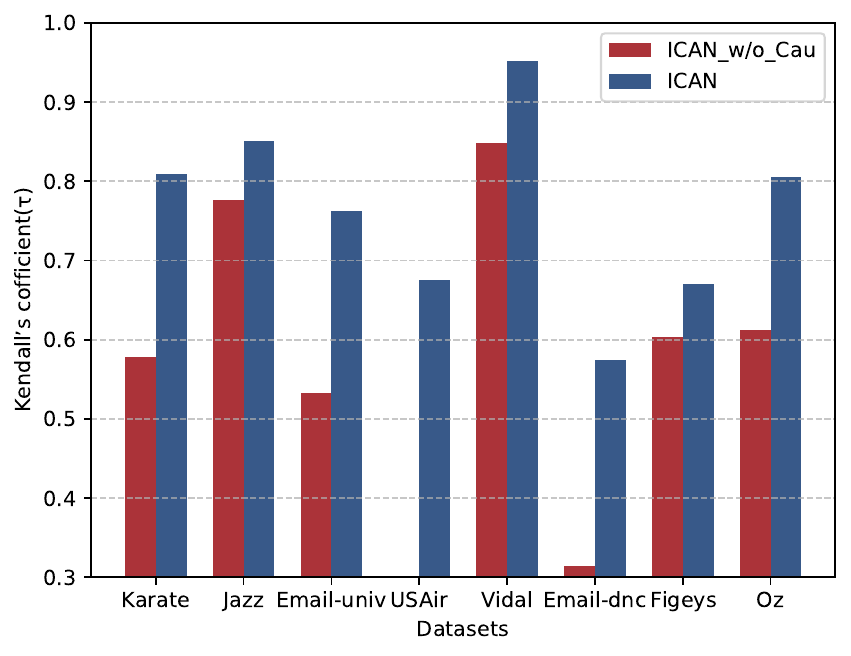}\label{wocausal} }
	\subfloat[
    \change{Ablation for co-optimization mechanism.}]{\includegraphics[width=.32\textwidth]{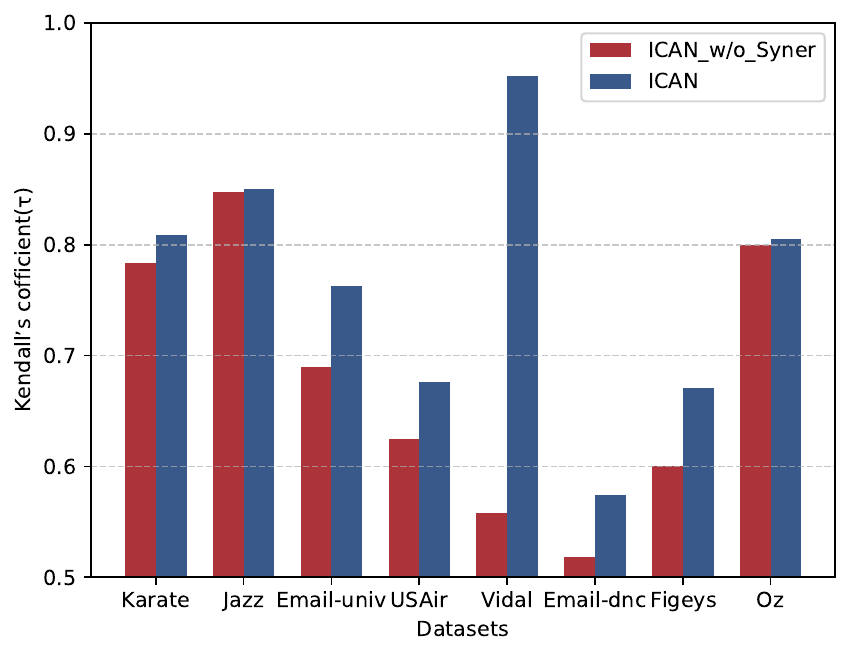} \label{Two-stage}}
    \subfloat[
    \change{Ablation for causal ranking loss.}]{\includegraphics[width=.32\textwidth]{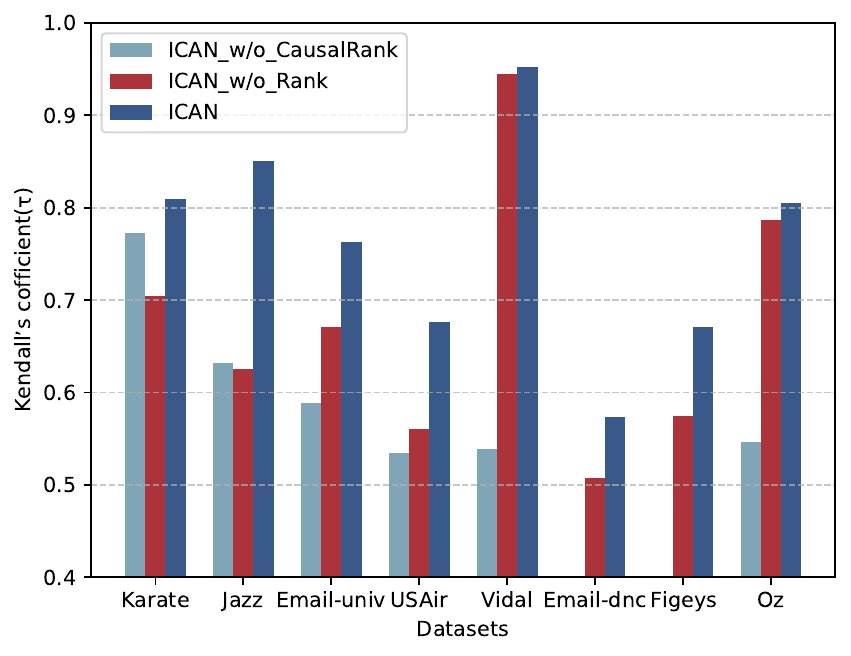}\label{wocausalrank} }
	\caption{\label{Ablation study results} Ablation study results.}
\end{figure*}

\begin{figure*}[htbp]
	\centering
	\subfloat[\change{Results with different $\lambda_1$.}]{\includegraphics[width=.32\textwidth]{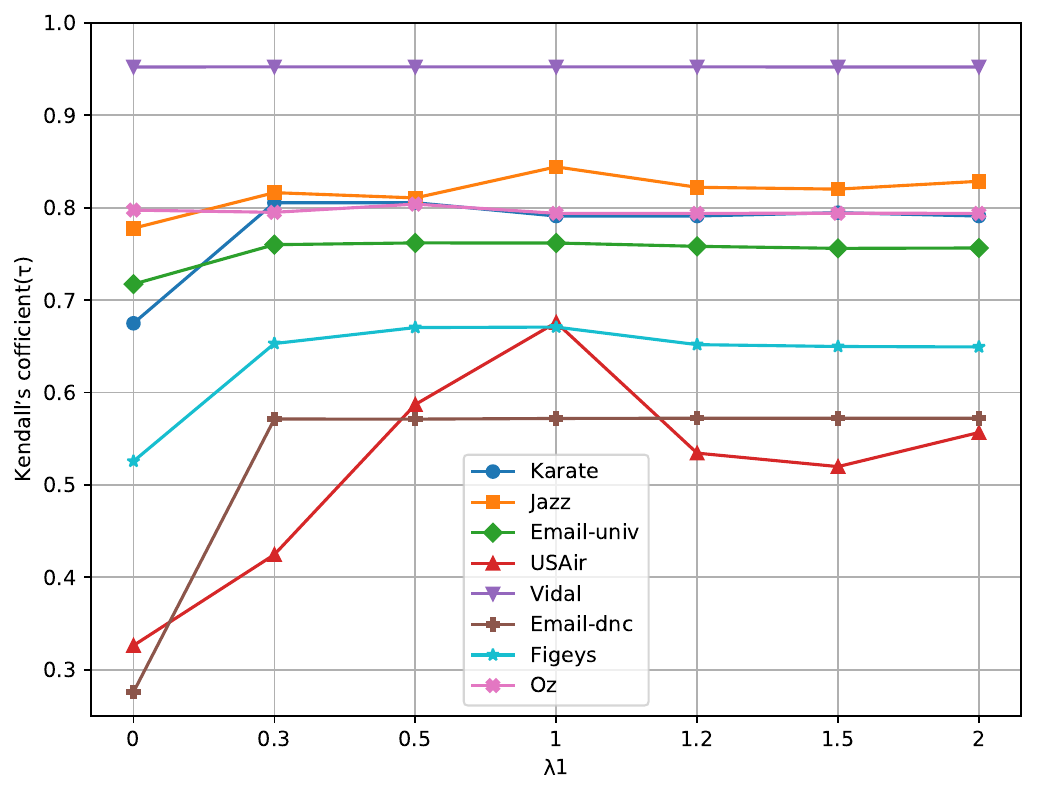}\label{sen_l1} }
	\subfloat[\change{Results with different $\lambda_2$.}]{\includegraphics[width=.32\textwidth]{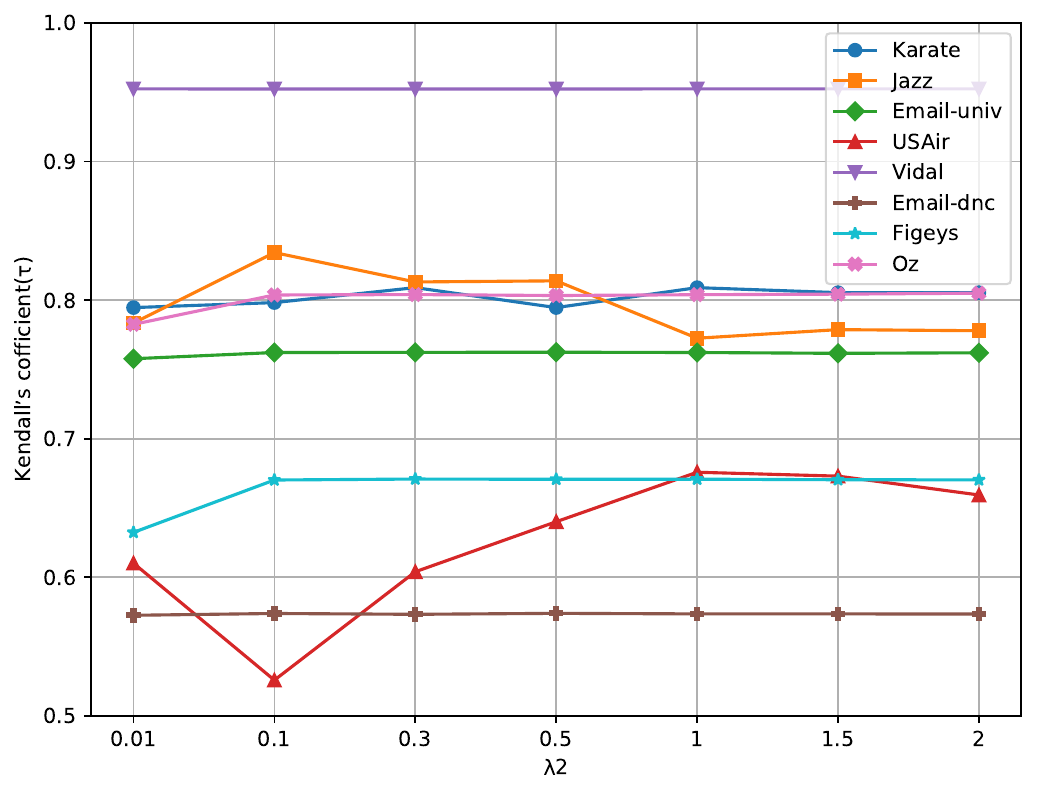} \label{sen_l2}}
	\subfloat[\change{Results with different $\lambda_3$.}]{\includegraphics[width=.32\textwidth]{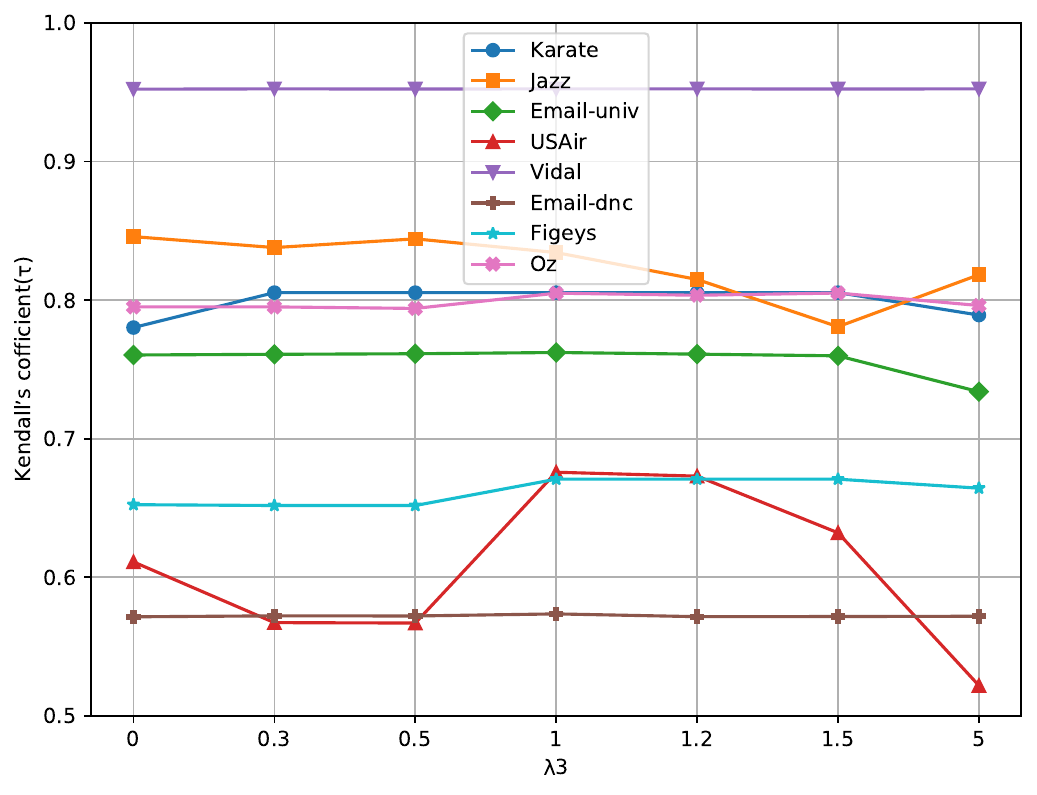} \label{sen_l3} }
	\caption{\label{sen_all}  Parameters sensitivity study on all the real-world networks for ICAN.}
\end{figure*}

\subsubsection{Ablation study for the influence-aware causal node embedding mechanism \label{RQ11}}  
As shown in the equations of $\mathcal{L}_1$ and $\mathcal{L}_2$, both are associated with $\bm{W}$, through which 
the node embeddings are learned to be causally related to the node influence scores. 
This indicates that the causal reconstruction loss and the causal ranking loss jointly contribute to the process of 
influence-aware causal representation learning. To emphasize the advantages of this influence-aware 
causal mechanism in ICAN, we conduct an ablation study by setting $\lambda_1 = 0$ and fixing $\bm{W}$ as an identity matrix. 
In this setup, the Markov blanket of the node influence score variable in the causal reconstruction loss is replaced by 
the low-dimensional node embeddings directly obtained from a traditional autoencoder. This variant can be regarded as a classical 
autoencoder applied to the ranking task and is denoted as ICAN\_w/o\_Cau. The comparison results between ICAN and ICAN\_w/o\_Cau 
across different networks are shown in Fig.~\ref{Ablation study results}(a). A noticeable performance decline of ICAN\_w/o\_Cau is 
observed in most node ranking tasks, highlighting the critical importance of the causal mechanism integrated into 
the autoencoder model.

\subsubsection{{Ablation study for the feature–task co-optimization mechanism}\label{RQ2}}
In this study, we integrate the causal ranking loss and causal reconstruction loss within a unified framework to enable causal feature extraction and ranking tasks. To evaluate the effectiveness of the proposed end-to-end 
 feature–task co-optimization mechanism, we conduct an ablation study by transforming the integrated model into a 
 two-stage implementation, in which the feature extraction and ranking processes are decoupled, thereby removing 
 the synergistic optimization effect. This variant is referred to as ICAN\_w/o\_Syner.
In the first stage, the causal reconstruction loss and a regularization term are employed as the loss function to optimize the causal node embedding module, yielding low-dimensional node representations that comprehensively encapsulate the topological properties of the network. During the second stage, the learned node embeddings are fed into the ranking prediction module, which is optimized using the causal ranking loss along with a regularization term to generate the predicted node influence score. These predictions are then compared against the node influence score derived from the SIR model, with Kendall’s $\tau$ coefficient used as the evaluation metric. The two-stage model is trained on the BA synthetic network and evaluated on \change{eight} real-world networks. The comparison results between ICAN and ICAN\_w/o\_Syner on different networks are presented in Fig.~\ref{Ablation study results}(b). Performance decline in ICAN\_w/o\_Syner is observed in most of the node ranking tasks. These results demonstrate that
ICAN facilitates mutual reinforcement between causal representation learning and ranking prediction, resulting in more accurate ranking outcomes.

\subsubsection{Ablation study for the causal ranking loss}\label{RQ3}
To validate the efficacy of the proposed the causal ranking loss, we conduct ablation experiments, replacing the CausalListMLE with the ListMLE loss for the node importance ranking, defined as ICAN\_w/o\_CausalRank, 
and additionally replacing the CausalListMLE with a MSE loss for the node importance score regression, 
denoted as ICAN\_w/o\_Rank. 
As illustrated in Fig.~\ref{Ablation study results}(c), the ICAN model consistently and significantly outperforms the ICAN\_w/o\_CausalRank and ICAN\_w/o\_Rank variants on all datasets. 
The superiority of the CausalListMLE loss stems from its inherent ability to capture ordinal relationships within ranked lists, which is more aligned with the nature of the node importance ranking objective. Overall, these findings conclusively demonstrate that integrating the causal ranking loss function is indispensable for achieving high-ranking accuracy, thereby affirming the design rationale behind our task-aware mechanism.

Remark: Our work follows a similar paradigm to the work of AGNN~\cite{xiong2024vital}
and STGCN~\cite{huang2025sparse}. Both frameworks achieve superior node importance ranking by integrating graph neural networks with autoencoder-based representation learning and transfer learning from synthetic to real-world networks. The novelty of our model compared to these two is that we introduce an influence-aware causal node embedding mechanism to ensure that the learned embeddings are causally related to node influence. We further propose a unified objective function that integrates a causal reconstruction loss and a causal ranking loss with a regularization term. This framework synergistically optimizes both feature representation and ranking, which is validated by ablation studies to enhance ranking performance and generalization capability.

\subsection{{Experiments with various training networks (RQ3)}} \label{RQ5}
To examine the dependency of the proposed model on training networks, we conduct experiments on three real-world networks—Karate, Jazz, and Email-univ—each trained under five distinct generative graph models (BA, ER, EH, QS, and RH). As shown in Tables~\ref{Karate test}, \ref{Jazz test} and \ref{Email-univ test}, ICAN consistently achieves the highest Kendall’s coefficient across all {generated training graphs}, demonstrating its {robust} transferability from synthetic to real-world networks.  Specifically, on the Karate network, ICAN attains an average Kendall’s coefficient of approximately 0.78, significantly surpassing all baseline methods. On the Jazz network, ICAN maintains robust and stable performance with average coefficients above 0.83, whereas competing methods exhibit substantial fluctuations across different generative networks. Similarly, on the Email-univ network, ICAN achieves the highest and most consistent results. These results collectively verify that ICAN exhibits superior flexibility and adaptability to variations in training graph structures compared with other existing approaches.

\begin{table}[htbp]
	\caption{Kendall’s coefficient of various methods on Karate network.}
	\label{Karate test}
	\centering
		\begin{tabular}{ccccccccc}
			\toprule
		Methods &BA & ER & EH  & QS & RH  \\
		\midrule
        GNN-Bet  &0.6530 &0.6137&0.6280&0.6280&0.5602\\
        GNN-Close&0.5352 &0.6887&0.6780&0.6244&0.6958\\
		RCNN &0.7528&0.7094&0.7130&0.7094&0.7058   \\
		CGNN &0.7058&0.6749&0.6749&0.6737&0.6737 \\
		AGNN &0.6505&0.7694&0.7586&0.6865&0.5828  \\
        STGCN &0.5640 &0.5568 &0.5604 &0.5928 &0.5568\\
        ICAN &\underline{0.8090}  &\underline{0.8092}&\underline{0.7910}&\underline{0.7622}&\underline{0.7442} \\
			\bottomrule
	\end{tabular}
\end{table}

\begin{table}[htbp]
	\caption{Kendall’s coefficient of various methods on Jazz network.}
	\label{Jazz test}
	\centering
		\begin{tabular}{ccccccccc}
			\toprule
			Methods &BA & ER & EH  & QS & RH  \\
			\midrule
            GNN-Bet  &0.6789 &0.6394&0.6711&0.6485&0.7178\\
            GNN-Close&0.6551 &0.2779&0.4140&0.5574&0.5806\\
			RCNN &0.8455&0.8014&0.8012&0.8156&0.8212   \\
			CGNN &0.7911&0.6214&0.6214&0.6219&0.6219 \\
			AGNN &0.7584&0.3610&0.7338&0.2400&0.2926  \\
            STGCN &0.7000 &0.6525 &0.6524&0.6673&0.6441\\
			ICAN  &\underline{0.8400}&\underline{0.8504}&\underline{0.8303}&\underline{0.8299}&\underline{0.8321} \\
			\bottomrule
	\end{tabular}
\end{table}

\begin{table}[htbp]
	\caption{Kendall’s coefficient of various methods on Email-univ network.}
	\label{Email-univ test}
	\centering
		\begin{tabular}{ccccccccc}
			\toprule
		Methods &BA & ER & EH  & QS & RH  \\
		\midrule
        GNN-Bet&0.3645 &0.3744&0.3646&0.4024&0.3614\\
        GNN-Close& 0.3629&0.4550&0.4662&0.5089&0.4616\\
		RCNN &0.7522 &0.8141&0.8149&\underline{0.7878}&\underline{0.7845}   \\
		CGNN &0.6488&0.6323&0.6325&0.6323&0.6326 \\
		AGNN &0.3100&0.3485&0.2995&0.2800&0.2889  \\
        STGCN &0.7327 &0.7141&0.7146&0.7038&0.7093\\
		ICAN &\underline{0.7625}  &\underline{0.8219}&\underline{0.8222} &0.7838&0.7737 \\
			\bottomrule
	\end{tabular}
\end{table}

From the results, we observe that the performance of ICAN exhibits slight variations when trained on different networks for a given target network. This behavior may stem from the structural similarities between the training and target networks. In future work, we plan to conduct a more systematic investigation of this phenomenon and establish a theoretical framework for selecting or generating suitable training networks for specific target networks.

\subsection{Parameter sensitivity analysis}

\change{To investigate the impact of the balanced parameters $\lambda_1$, $\lambda_2$ and $\lambda_3$, we conduct a parameter sensitivity analysis of the models trained on the BA network, evaluating their performance on eight real-world networks.}
Let $\lambda_2 = \lambda_3 =1$, 
the Kendall’s coefficients with different values of  $\lambda_1$
are shown in 
Fig.~\ref{sen_all}(a). 
Similarly, setting the values of the other two parameters to be the same, 
the results with \change{different} values of $\lambda_2$ and $\lambda_3$ are displayed in 
Fig.~\ref{sen_all}(b) and Fig.~\ref{sen_all}(c)
respectively. 
\change{The results demonstrate that while the model responds to parameter variations, it exhibits significant robustness and consistency across the majority of datasets.
Specifically, 
for $\lambda_1$, the performance remains stable and optimal within the range $(0.5, 1.2)$ for all networks.
Similarly, for $\lambda_2$ and $\lambda_3$, most networks (7 out of 8) achieve optimal or near-optimal performance when the parameters are set around $1$.  The Jazz network is the only exception, requiring lower values ($\lambda_2 \approx 0.1, \lambda_3 \approx 0.5$) to reach its peak, likely due to its distinct topological structure compared to the other networks.
In practical scenarios without ground truth labels, users need a reliable default setting. Based on the extensive overlap of optimal regions across the majority of datasets, we empirically recommend a unified setting of $(\lambda_1, \lambda_2, \lambda_3) = (1, 1, 1)$. This configuration effectively balances the model components and ensures robust performance for most real-world networks, eliminating the need for extensive parameter tuning in most cases.}

\change{To investigate the effect of different parameterizations for synthetic training networks,} 
\change{we generated BA networks with varying number of nodes (1,000, 2,000, and 3,000) and average degrees (2, 4, and 8). Their performance is evaluated on six real-world datasets, as shown in 
Fig.~\ref{3D}. 
The results indicate that the model yields the best performance 
with 1,000 nodes and an average degree of 4. Based on these findings, we recommend this parameter as the default training setting for practical applications.} 


\begin{figure*}[htbp]
	\centering
	\subfloat[\change{Karate}]{\includegraphics[width=.32\textwidth]{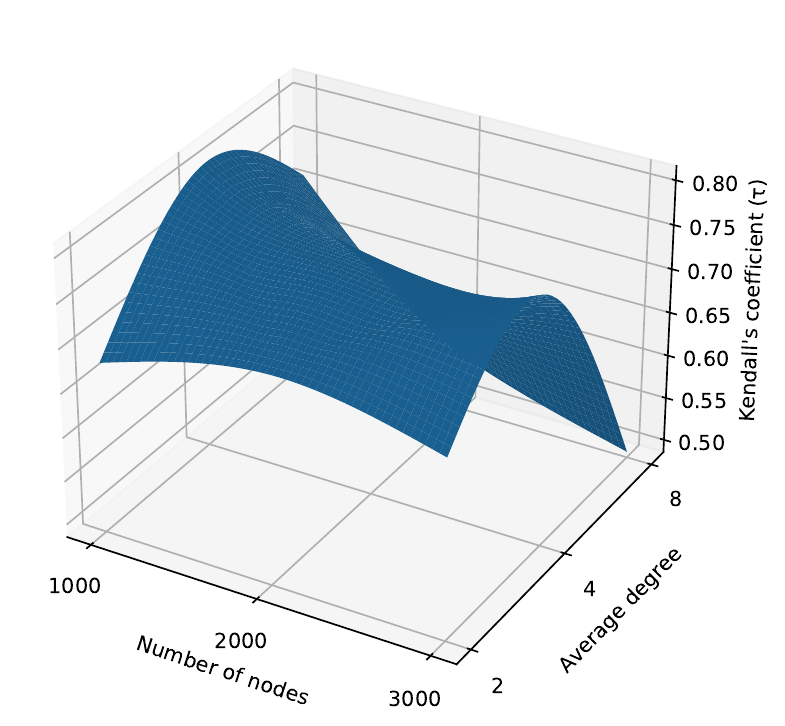}\label{Karate} }
	\subfloat[\change{Jazz}]{\includegraphics[width=.32\textwidth]{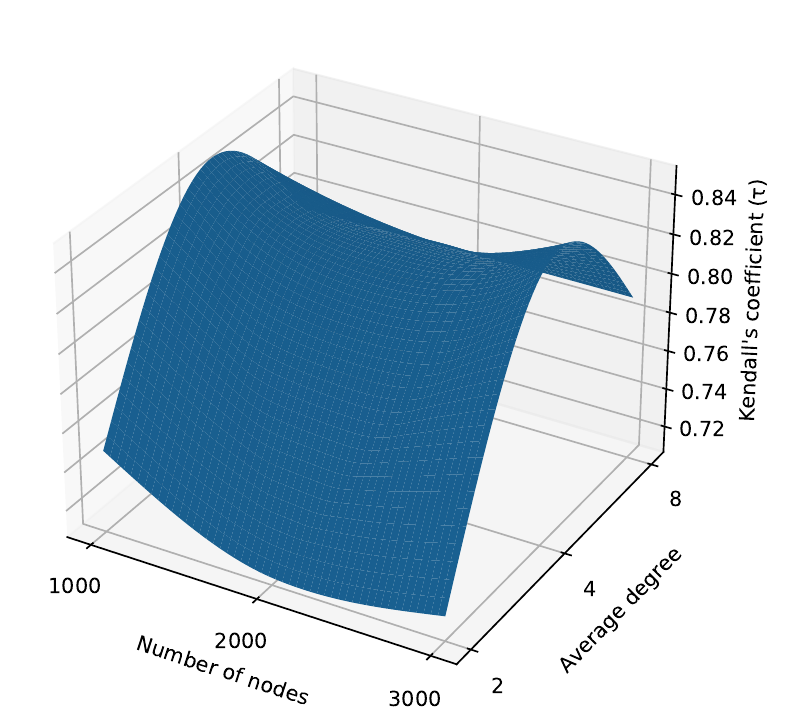} \label{Jazz}}
	\subfloat[\change{Email-univ}]{\includegraphics[width=.32\textwidth]{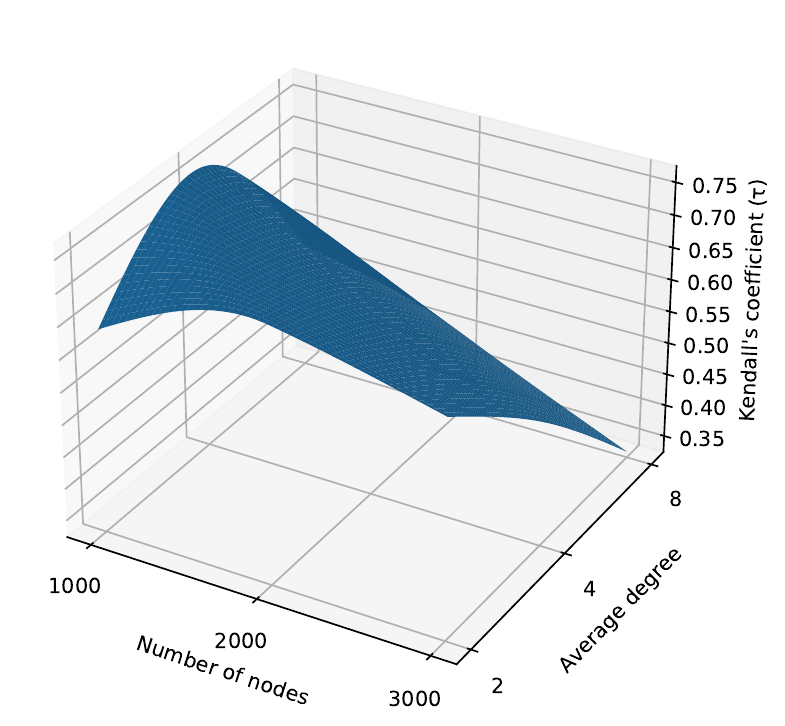} \label{Email-univ} }%
    \\
    \subfloat[\change{USAir}]{\includegraphics[width=.32\textwidth]{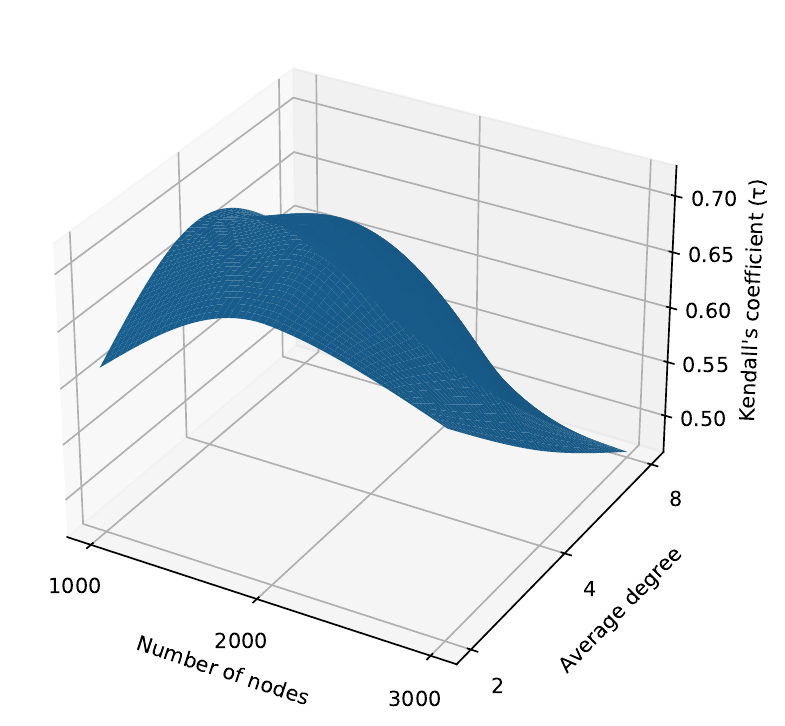}\label{USAir} }
    \subfloat[\change{Figeys}]{\includegraphics[width=.32\textwidth]{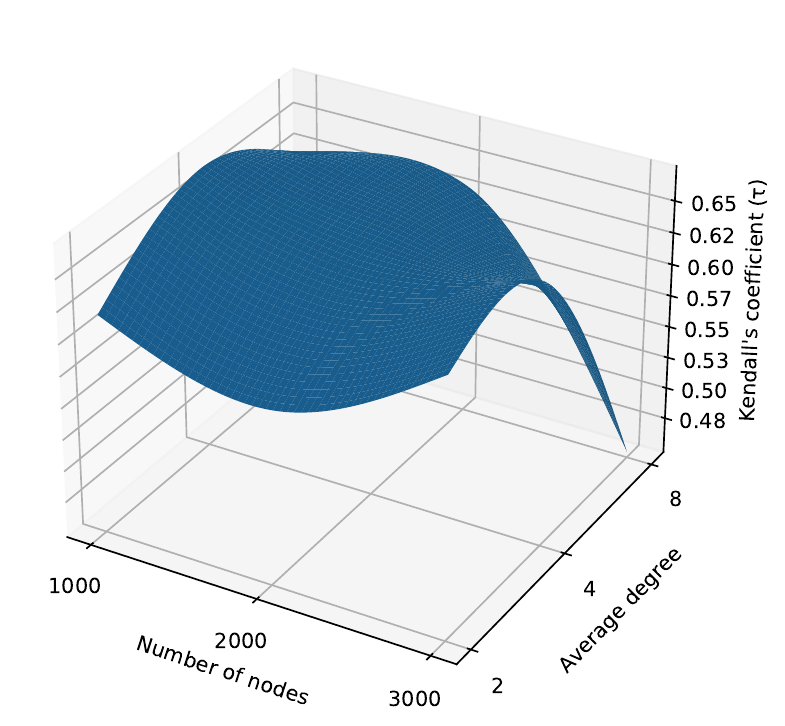}\label{Figeys} }
    \subfloat[\change{Oz}]{\includegraphics[width=.32\textwidth]{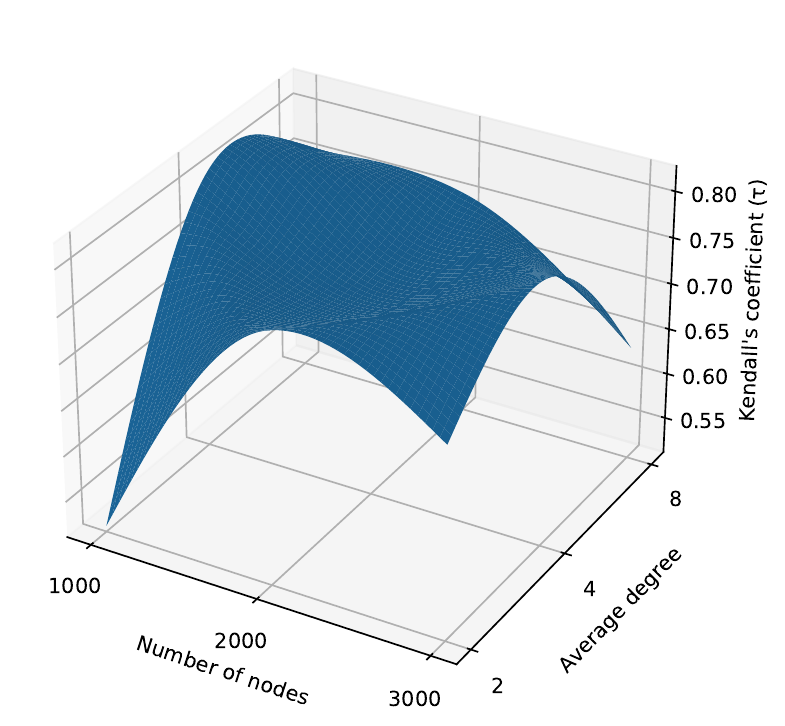}\label{Oz} }
	\caption{\label{3D}  \change{Performance comparison of different training network parameters (numbers of nodes and average degrees) on six
real-world networks.}}
\end{figure*}

\section{Discussion}
In this section, we discuss the implications of the proposed model from both network science and engineering perspectives. First, we examine the degree distribution of the most influential nodes identified by the model, providing insights into the structural characteristics that underlie node importance. We then explore the model’s behavior and applicability in \change{real-world networks with structural descriptors as node features and }social networks with node attributes, highlighting its potential in more complex, attribute-rich environments.
\subsection{Degree distribution of top 10\% influential nodes}

In this subsection, we analyze the degree distribution of nodes identified as important according to the predicted ranking. 
Figs.~\ref{degree distribution1}--\ref{degree distribution3}
compare the degree distribution of the three networks and the corresponding influence scores of the top 10\% nodes identified by ICAN, DC, and CGNN. Each subfigure combines a blue histogram representing the network’s degree distribution (left $y$-axis, relative frequency) with red scatter points indicating the influence scores of individual nodes (right $y$-axis, node influence score).

A clear methodological contrast emerges. As expected for the degree–centrality-based approach, the results for DC in
Figs.~\ref{degree distribution1}(b)--\ref{degree distribution3}(b) 
exhibit a strong positive correlation: nodes with the highest influence scores 
coincide with those of the largest degrees, concentrating red points at the high-degree end of the spectrum. 
In contrast, ICAN and CGNN reveal markedly different patterns. High-scoring nodes appear across a broader degree range, with a noticeable concentration in the moderate-degree region. This suggests that these methods attribute high influence to nodes that are not necessarily the network’s primary hubs.

This divergence highlights a key insight: ICAN captures influence patterns that extend beyond simple degree centrality. In network dismantling scenarios, where intervention costs vary across 
nodes (for example, high-degree nodes may entail greater operational or socio-economic expenses), ICAN’s emphasis on structurally balanced, propagation-relevant nodes results in a prioritized degree distribution conducive to 
lower-cost, targeted interventions. These findings point to practical strategies for resource-constrained mitigation, including the identification of critical components for the efficient dismantling of non-cooperative networks.

\begin{figure*}[htbp]
	\centering
	\subfloat[ICAN.]{\includegraphics[width=.32\textwidth]{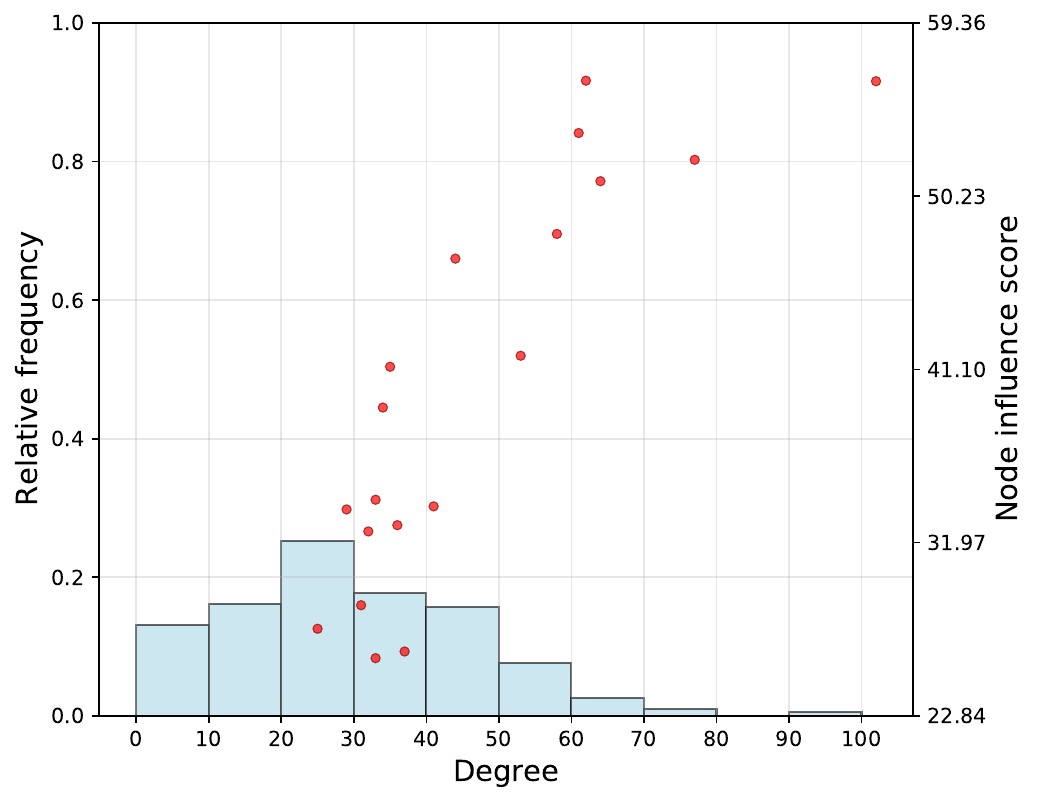}\label{ICAN jazz} }
	\subfloat[DC.]{\includegraphics[width=.32\textwidth]{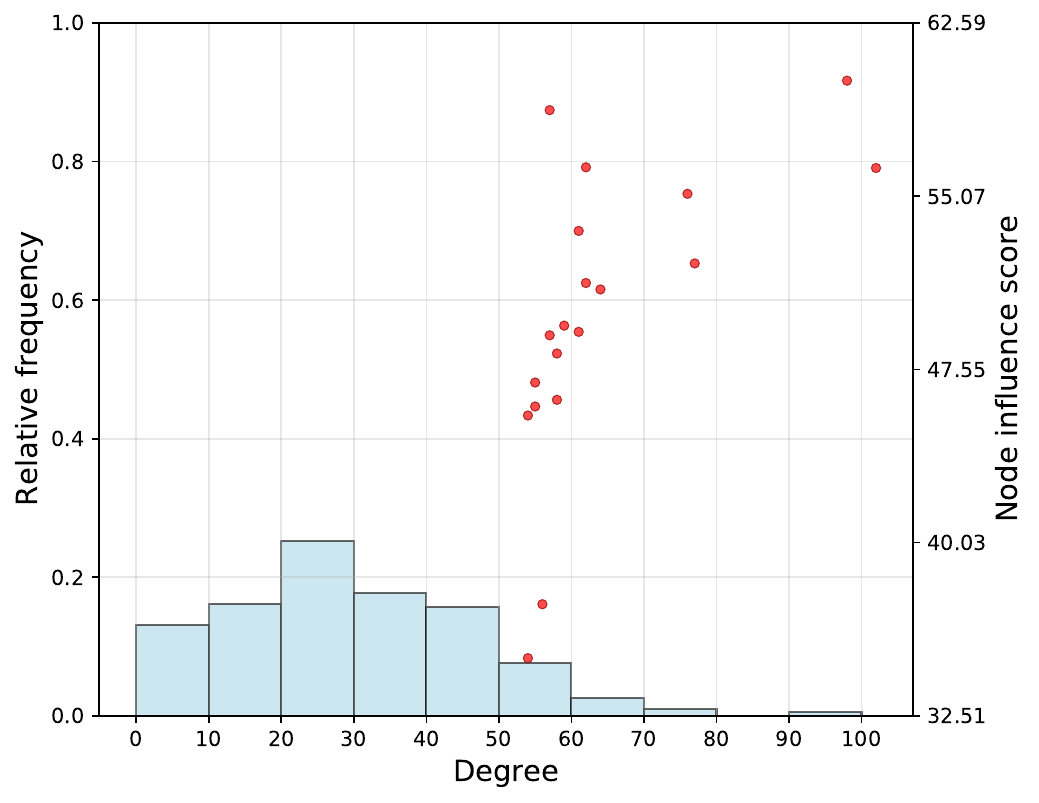} \label{DC jazz}}
	\subfloat[CGNN.]{\includegraphics[width=.32\textwidth]{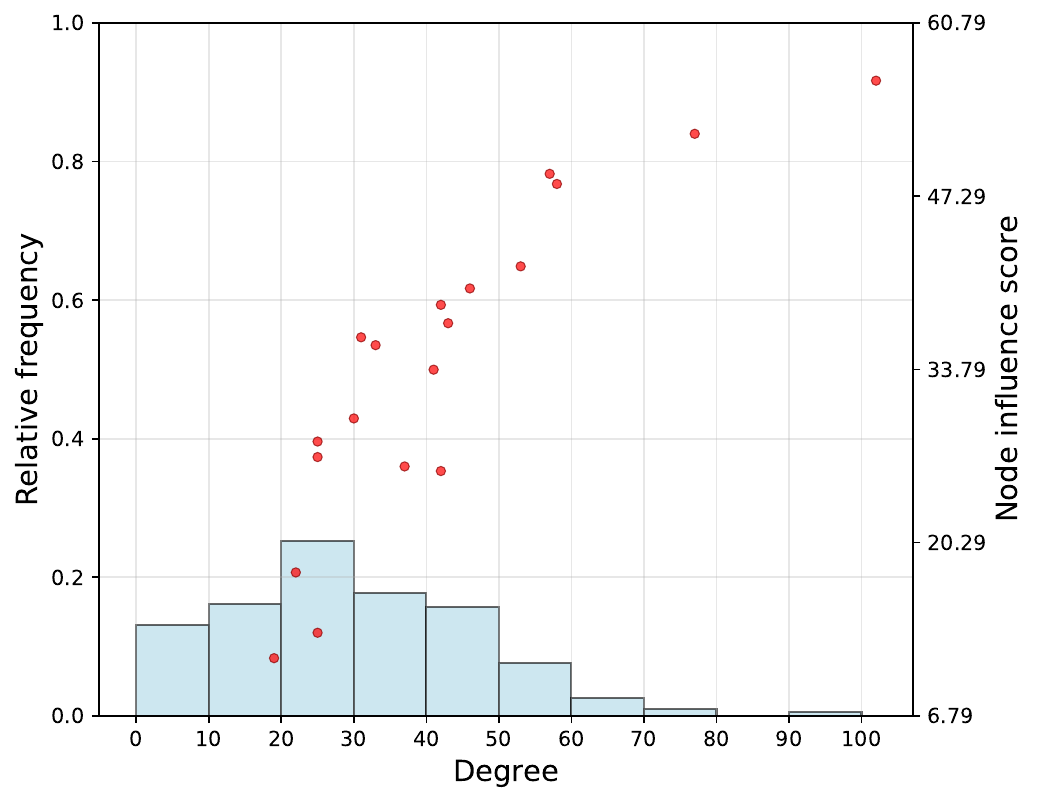}\label{CGNN jazz} }
	\caption{\label{degree distribution1} Degree distribution of the detected top $10\%$ important nodes in Jazz network.}
\end{figure*}

\begin{figure*}[htbp]
	\centering
	\subfloat[ICAN.]{\includegraphics[width=.32\textwidth]{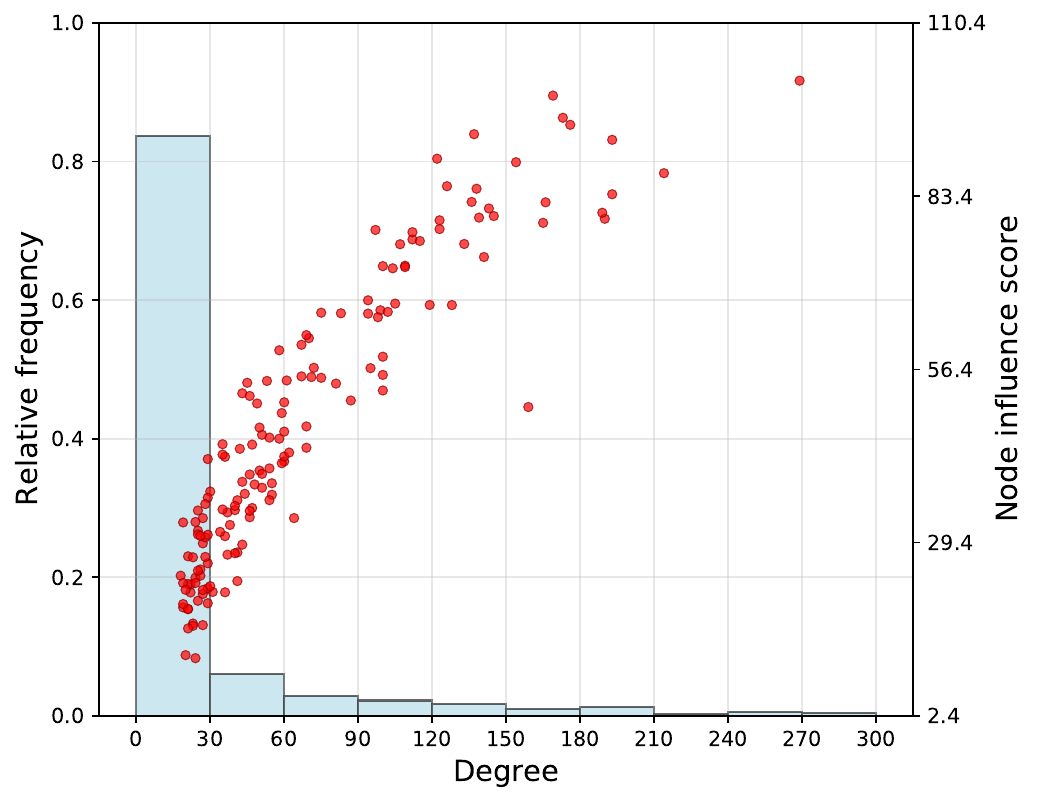}\label{ICAN USAir} }
	\subfloat[DC.]{\includegraphics[width=.32\textwidth]{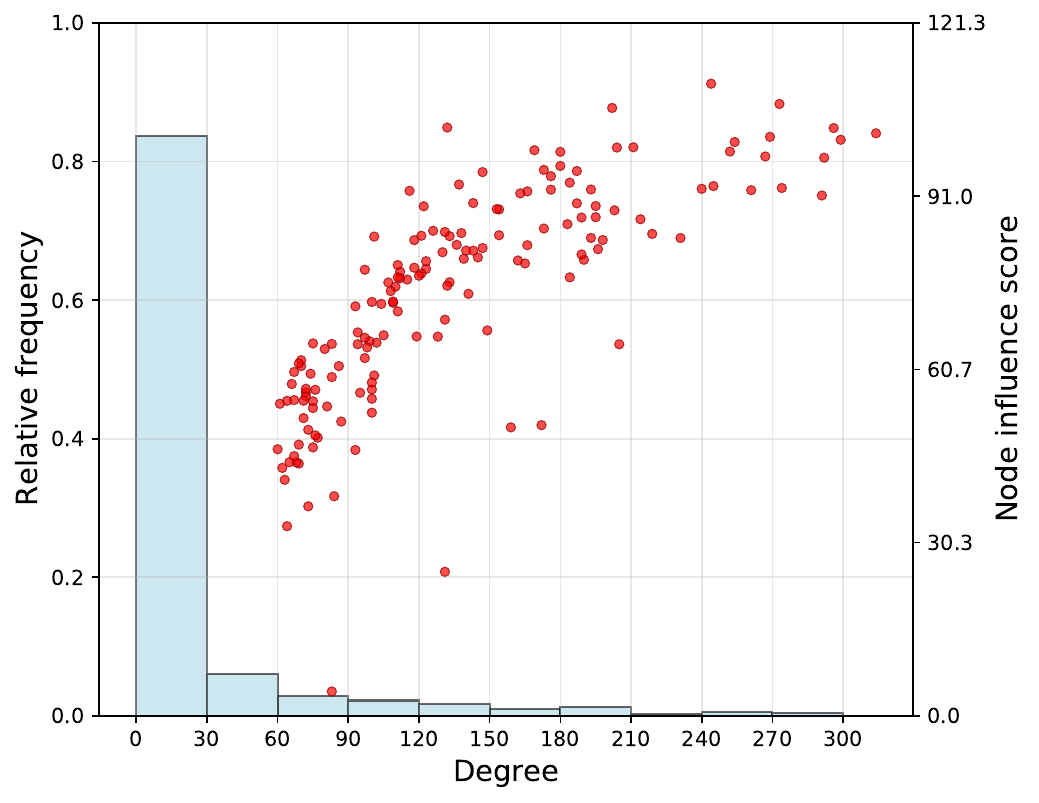} \label{DC USAir}}
	\subfloat[CGNN.]{\includegraphics[width=.32\textwidth]{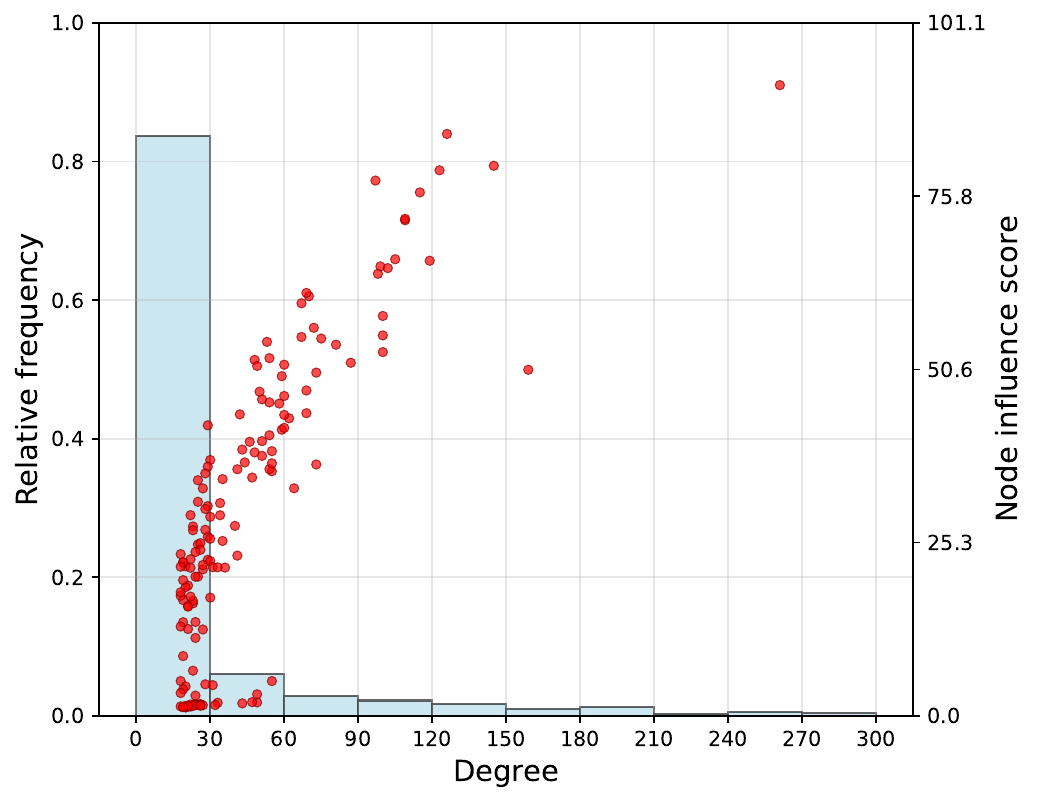}\label{CGNN USAir} }
	\caption{\label{degree distribution2} Degree distribution of the detected top $10\%$ important nodes in USAir network.}
\end{figure*}

\begin{figure*}[htbp]
	\centering	\subfloat[ICAN.]{\includegraphics[width=.32\textwidth]{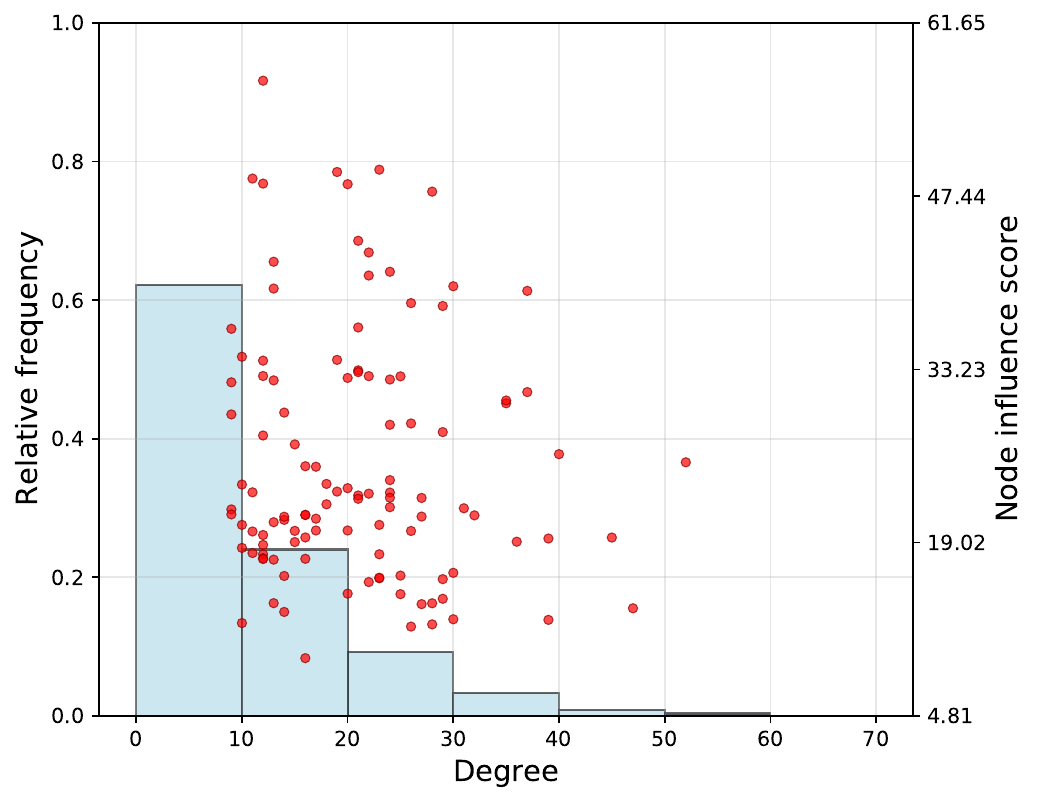}\label{ICAN email-univ} }
	\subfloat[DC.]{\includegraphics[width=.32\textwidth]{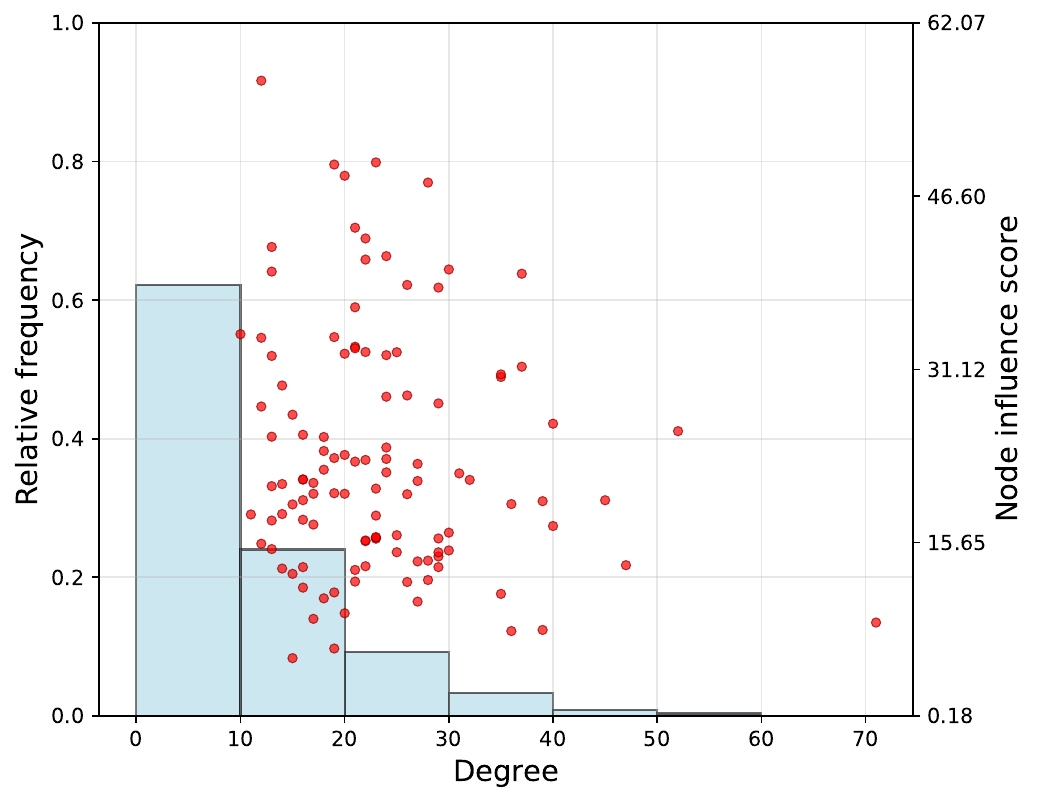} \label{DC email-univ}}
	\subfloat[CGNN.]{\includegraphics[width=.32\textwidth]{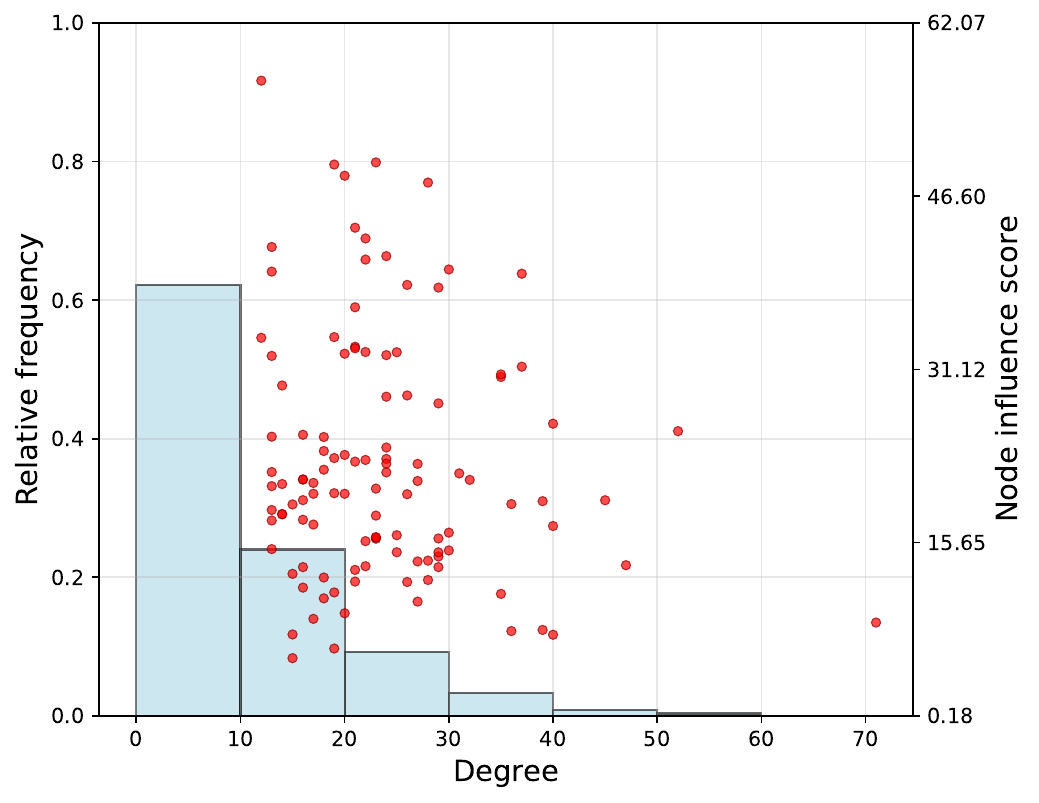}\label{CGNN email-univ} }
	\caption{\label{degree distribution3} Degree distribution of the detected top $10\%$ important nodes in Email-univ network.}
\end{figure*}


\subsection{\change{Case study}}
\change{Although ICAN utilizes node2vec embeddings as the feature matrix in our primary experiments, it is capable of flexibly accommodating diverse feature inputs. 
Specifically, for non-attributed networks, it can 
leverage explicit topological descriptors, while for attributed networks, it can directly utilize the raw attribute matrix. To substantiate this flexibility, we conduct two 
supplementary experiments in this subsection.}
\subsubsection{\change{Experiments with structural descriptors as node features}}
\change{In this experiment, we select six classical node centrality measures
 as feature vectors to represent the characteristics
and structural properties of nodes, including degree centrality, betweenness centrality, closeness centrality, clustering coefficient, eigenvector centrality, and load centrality. To reduce the risk of overfitting, feature values 
for each node are normalized using min-max scaling
so that they fall within the interval $(-1,1)$. The resulting normalized centrality vectors form the feature matrix fed into ICAN. We train ICAN on BA synthetic network and evaluate its generalization performance on eight real-world networks. As shown in Table~\ref{ICAN_C}, the variant using structural descriptors as node features (ICAN-S) achieves slightly lower average ranking performance than the node2vec-based version (ICAN)}. 
\change{However, as reported in Table~\ref{BA training}, ICAN-S consistently achieves superior average ranking performance over all other baselines across datasets.
These results indicate that ICAN can flexibly adapt to different types of input features while maintaining strong generalization and robust ranking performance.}

\begin{table*}[htbp]
    \caption{\change{kendall's $\tau$ of ICAN and ICAN-S on eight real-world networks.}}
    \label{ICAN_C}
    \centering
    \resizebox{\linewidth}{!}{ 
    \begin{tabular}{cccccccccc}
    \toprule
       Methods  & Karate & Jazz & Email-univ & USAir &Vidal &Email-dnc &Figeys &Oz & Average \\
        \midrule
       ICAN &\underline{0.8090} &\underline{0.8504} & \underline{0.7625} & \underline{0.6758} &\underline{0.9524} & 0.5740 &\underline{0.6708} &\underline{0.8051} & \underline{0.7625} \\
       ICAN-S &0.8033 & 0.8259 & 0.6598 & 0.6338 & 0.9524 &\underline{0.5939} & 0.6391 & 0.7819 & 0.7363\\
       \bottomrule
    \end{tabular}}
\end{table*}

\subsubsection{\change{Social network with attribute-enriched node features}}
In this case study, we show that ICAN can be applied to a social network with attribute information. The feature matrix can be constructed from node attributes. Cora~\cite{cabanes2012cora} is a citation network where nodes represent scientific publications and edges represent directed citation links between them. Each node feature in the Cora dataset is a 1,433-dimensional binary bag-of-words vector representing the presence or absence of specific terms in the paper’s content.
Our model is trained on BA synthetic networks (using \change{node2vec} for feature generation)  
and evaluated on the attributed Cora real-world network. The experimental results, as shown in Fig.~\ref{with attribute cora}, demonstrate that our method achieves superior ranking performance compared to all baseline approaches.  
The results demonstrate that ICAN can be applied to social networks with attribute information, indicating its potential for broader applications in social network analysis. However, in such attribute-rich networks, 
inconsistencies may arise between topological structures and node attributes~\cite{ma2025node}. 
Therefore, future work will focus on developing more effective embedding strategies that can integrate both sources of information in a 
coherent manner.


\begin{figure}[htbp]
	\centerline{
		\includegraphics[scale=0.3]{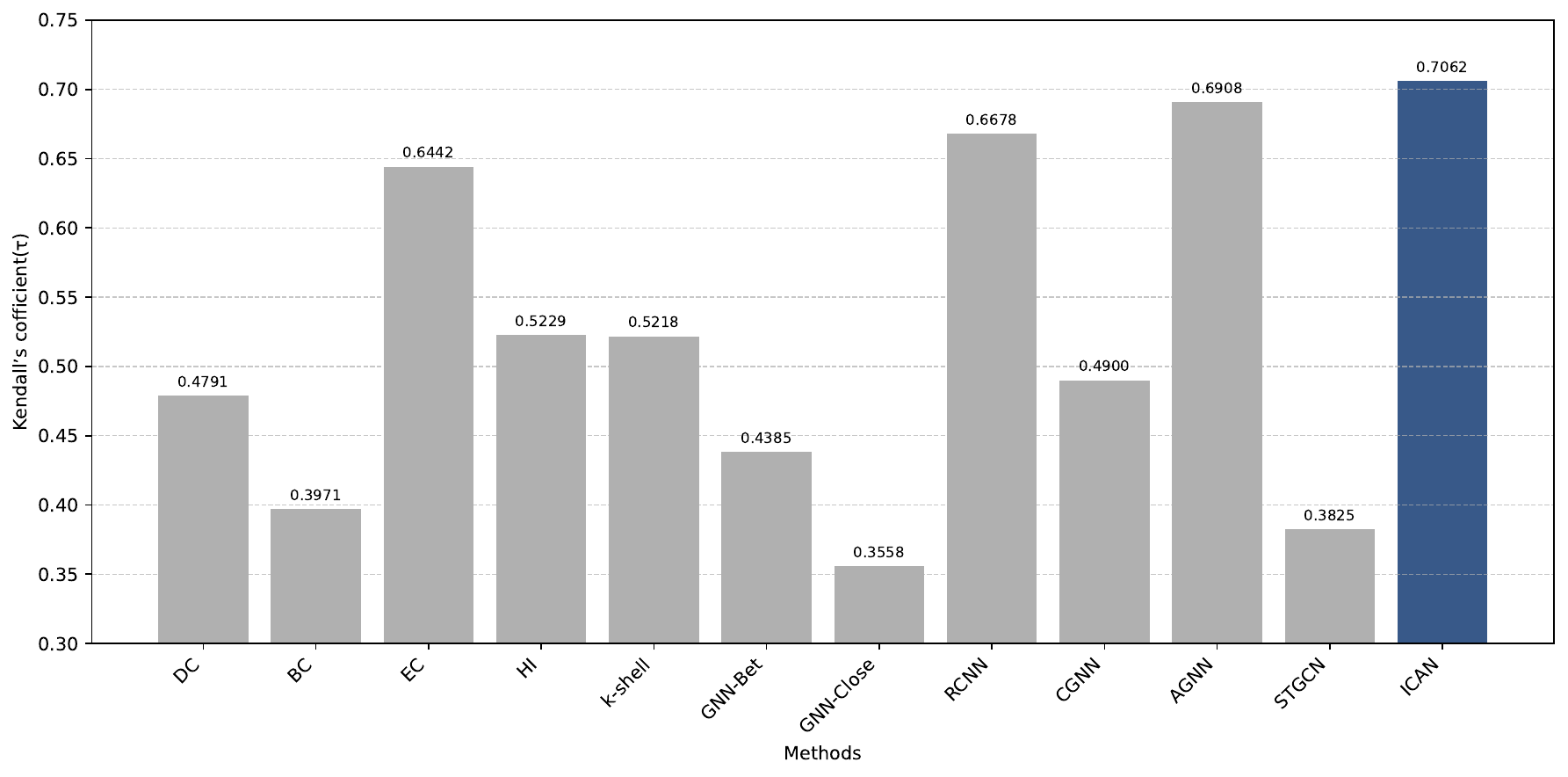}}
	\caption{Kendall's $\tau$  of different methods on Cora dataset. }
	\label{with attribute cora}
\end{figure}

\section{Conclusions}\label{sec5}
Quantifying node importance is a challenging task, particularly in non-cooperative scenarios where the target network structure is \change{inaccessible} in advance due to privacy or security constraints. 
This paper proposes a node importance ranking method for complex networks based on an influence-aware causal autoencoder model. The proposed method effectively captures the network-invariant causal relationships between 
node embeddings and influence scores. This design allows ICAN to learn robust, low-dimensional node representations solely from 
synthetic networks, which can then be applied to node-ranking tasks on various real-world networks. 
We conduct comprehensive experiments by training ICAN on different types of synthetic networks and evaluating it on 
diverse benchmark networks, comparing its performance with eleven representative node-ranking methods. 
The results confirm that ICAN achieves state-of-the-art performance, excelling in both ranking accuracy and 
cross-network generalization.
The results of the ablation studies confirm that both the influence-aware causal node embedding mechanism and the feature–task co-optimization strategy—based on the jointly defined causal ranking loss and causal 
reconstruction loss—play crucial roles in ensuring that the learned representations generalize effectively across 
diverse target graphs and are better suited for downstream ranking tasks, thereby substantially enhancing the model’s overall performance. 
Future work will focus on establishing a theoretical framework for selecting or generating suitable training networks for a given network engineering task. In addition, subsequent research will undertake a 
deeper investigation of ICAN in practical network engineering contexts, such as developing cost-aware strategies for network dismantling and designing embedding methods that effectively integrate topological and semantic 
information in social networks.

\bibliographystyle{IEEEtran}
\bibliography{reference}

\end{document}